\documentclass[final,3p,12pt,authoryear,a4paper]{elsarticle}

\usepackage{float}
\usepackage{setspace}
\usepackage[none]{hyphenat}
\usepackage{amsmath,amsfonts,amssymb,amsthm,pxfonts,dsfont}
\usepackage{natbib}
\usepackage[bookmarks, colorlinks]{hyperref}
\hypersetup{linkcolor=blue,citecolor=blue,filecolor=black,urlcolor=blue} 
\usepackage{graphicx} 
\usepackage{color}
\usepackage{geometry}
\usepackage{multirow}
\usepackage{booktabs}
\usepackage{rotating} 
\usepackage{subfig}
\usepackage[american]{babel}
\usepackage{longtable,lscape,ltcaption}
\usepackage{rotating}
\usepackage{soul}
\usepackage[normalem]{ulem}
\usepackage{pdfpages}

\emergencystretch=\maxdimen
\DeclareMathOperator{\diag}{diag}

\begin{document}

\biboptions{longnamesfirst}
\newtheorem{corollary}{Corollary}

\journal{}

\begin{frontmatter}

\title{Time-Varying Identification\\ of Structural Vector Autoregressions}

\author[eur]{Annika Camehl\corref{cor1}}
\author[um]{Tomasz Wo\'zniak}

\cortext[cor1]{Corresponding author. \emph{Email address:} \href{mailto:camehl@ese.eur.nl}{camehl@ese.eur.nl}.}

\address[eur]{Erasmus University Rotterdam, Burgemeester Oudlaan 50, 3062 PA Rotterdam, The Netherlands} 
\address[um]{University of Melbourne, 111 Barry St., 3053 Carlton, VIC, Australia}

\begin{abstract}
\noindent We propose a novel Bayesian heteroskedastic Markov-switching structural vector autoregression with data-driven time-varying identification. The model selects among alternative patterns of exclusion restrictions to identify structural shocks within the Markov process regimes. We implement the selection through a~multinomial prior distribution over these patterns, which is a~spike'n'slab prior for individual parameters. By combining a Markov-switching structural matrix with heteroskedastic structural shocks following a stochastic volatility process, the model enables shock identification through time-varying volatility within a regime. 
As a~result, the exclusion restrictions become over-identifying, and their selection is driven by the signal from the data. Our empirical application shows that data support time variation in the US monetary policy shock identification. We also verify that time-varying volatility identifies the monetary policy shock within the regimes.
\end{abstract}


\end{frontmatter}

\newpage

\section{Introduction}\label{sec:Introduction}

\noindent Most researchers assume that restrictions identifying shocks in structural vector autoregressions remain constant over time, as proposed by \cite{sims1980}. This time-invariance in identification is imposed in the literature regardless of the identification strategy or time variation in the parameter values \citep[see, e.g.,][]{Primiceri2005,Sims2006}. However, time-invariant identification strategies might be at odds with the data. We relax this time-invariance assumption and introduce a novel approach  that implements time-varying identification (TVI) through a data-driven search mechanism.

We propose a new model, a Bayesian Markov-switching (MS) structural vector autoregression (SVAR) with time-varying identification and heteroskedasticity following stochastic volatility (SV). This framework allows structural shock identification to adapt to the data over time. Our paper has three contributions. As the first contribution, we introduce time-varying identification through a data-driven search. This allows our model to select alternative exclusion restrictions across regimes. Two key features of our model implement this selection: first, we introduce a new prior distribution to estimate a TVI indicator which alters the identification pattern in each MS regime. Second, we let the structural matrix be regime-dependent as in \cite{Sims2006}. Together, these features let both the structural parameters as well as different exclusion restrictions vary over regimes. We do not rule out time-invariant identification as no restriction pattern is associated with a specific regime. Instead, we rely on the data to select which exclusion restrictions are supported within a regime. We implement the~data-driven search process for regime-specific identifying exclusion restrictions for a specific shock, hence, providing partial identification. 

As the second contribution, we introduce identification through heteroskedasticity \textit{within} regimes. The data-driven search of regime-dependent exclusion restriction patterns is only feasible if these restrictions are overidentifying within the regime. We achieve this by partially identifying a particular structural shock through heteroskedasticity within a regime. To that end, we rely on three features combining them in a novel way in our model. First, we use that our structural matrix follows a MS process while the structural shocks' variances instead evolve according to a SV process. That way, the conditional variances can change at each period and shocks can thus be heteroskedastic within a regime. 
Second, we adopt a non-centered representation of the SV process as in \cite{Kastner2014} in our model and let the SV prior distribution concentrate on a homoscedastic model. This specification ensures that identification through heteroskedasticity is driven by the data rather than by the prior. Moreover, the non-centered representation simplifies the process of verifying  heteroskedasticity, reducing it to a single parameter. This allows us to introduce as the third novel feature, an easy way to confirm within-regime heteroskedasticity. Specifically, we define partial identification through heteroskedasticity within regimes extending the results in \cite{LSUW2022} who use SVARs without MS. We then verify identification in each regime by checking that a shock-specific regime-dependent log-volatility standard deviation parameter implies heteroskedastic shocks within a regime. As a new feature, we allow this volatility parameter to be regime-dependent. That way, confirming within-regime heteroskedasticity reduces to an easy check: verifying that this parameter is different from zero.

As the third contribution, we provide new empirical insights by showing that the data strongly support time variation in US monetary policy shock identification. We show that these shocks are partially identified through heteroskedasticity. Our results are based on a model with industrial production growth, consumer price index inflation, federal funds rates, term spreads, M2 money, and stock prices for January 1960 to September 2024. We include four alternative restriction patterns which allow interest rates to move simultaneously with (1) inflation and output \citep{leeperWhatDoesMonetary1996a}, (2) inflation, output, and term spread  \citep{Vazquez2013,diebold2006}, (3) inflation, output, and money \citep{Andres2006,BelongiaIreland2015}, and (4) only money \citep{LeeperRoush2003,Sims2006}. Additionally, we allow for two alternatives in the fourth row: no restriction or a zero on the simultaneous movement of term spread with interest rates, implying that monetary policy interventions narrow the term spread while the short term interest rate does not move \citep{Baumeister2013,Feldkircher2016}. 

Our results show that in the first regime, which dominates before 2000 and in more tranquil economic conditions, data strongly favor identifying shocks based on simultaneous movements in interest rates, inflation, output, and term spreads. This suggests that the monetary policy shock is best identified in relation to the entire bond yield curve, reflecting adjusted expectations on future economic developments in spirit of \cite{diebold2006}. In this regime, an expansionary shock increases inflation, stimulates long-term economic growth, flattens the yield curve, features a liquidity effect, and raises stock prices.  In contrast, the second regime, which mostly occurs after 2000, particularly during the global financial crisis, the COVID-19 pandemic, and periods of unconventional monetary policy, supports identifying monetary policy shocks through simultaneous movements in interest rates, inflation, output, and money supply. This aligns with the results in \cite{BelongiaIreland2015} who show that after 2008, 
 monetary policy shocks are best identified by the high volatility of monetary aggregates. In this regime, a shock that lowers interest rates significantly stimulates economic activity after a year, increases the money supply, and steepens the bond yield curve.

Our model contributes to the literature on SVAR identification by introducing data-driven, regime-dependent identification patterns which extends existing methods that impose such patterns \textit{a priori}.  For example, \cite{KimuraNakajima2016} analyse changes in Japan's monetary policy using two exogeneously given regimes, based on interest rates or bank reserves, embedded within a latent threshold time-varying parameter SVAR.
\cite{Bacchiocchi2017} impose two identification patterns with an additive relationship in a heteroskedastic SVAR and apply them to study the effects of monetary policy. They report differences in impulse responses during the Great Moderation compared to previous periods.
\cite{Arias2024} impose sign restrictions on impulse responses to identify US monetary policy shocks and add restrictions on contemporaneous structural parameters in \textit{a priori}-defined periods when the Fed used the federal funds rates as the main policy instrument. \cite{Pagliari2024} identifies unconventional monetary policy shocks of the ECB by restricting the expectation component spread's reaction to the shock to be positive before and zero after June 2014. Our approach challenges the common assumption of time-invariant shock identification, allows the data to decide on the identification pattern for various periods, and estimates the probabilities of regime occurrences. Moreover, our paper relates to the literature acknowledging uncertainty around identifying restrictions in SVAR models \citep[such as][]{BaumeisterHamilton2018,GiacominiKitagawaKitagawa2022} looking at it through the angle of uncertainty over time.

In the following, we introduce our new model in Section \ref{sec:Model}, and report the empirical evidence on time-varying identification of US monetary policy shocks in Section \ref{sec:resultsTVI}.

\section{Data-driven time-varying identification}\label{sec:Model}



\noindent To study time variation in the identification of structural shocks, we estimate the structural vector autoregressive model:
\begin{align}
	\mathbf{y}_t &= \sum_{l=1}^{p} \mathbf{A}_l \mathbf{y}_{t-l} + \mathbf{A}_d \mathbf{d}_t + \boldsymbol{\varepsilon}_t  \label{eq:rf} \\
\mathbf{B}\left(s_t, \boldsymbol{\kappa}(s_t)\right) \boldsymbol{\varepsilon}_t &= \mathbf{u}_t \label{eq:sf}
\end{align}
where the vector $\mathbf{y}_t$ of $N$ dependent variables at time $t$ follows a vector autoregressive model in equation~\eqref{eq:rf} with $p$ lags, a $d$-vector of deterministic terms $\mathbf{d}_t$, an error term $\boldsymbol{\varepsilon}_t$, and $N\times N$ autoregressive matrices $\mathbf{A}_l$, for lag $l = 1, \dots, p$, and the $N\times d$ matrix $\mathbf{A}_d$ of deterministic term slopes. The structural equation~\eqref{eq:sf} links the reduced form residuals to the structural shocks $\mathbf{u}_t$ in a time-varying linear relationship through the $N\times N$ structural matrix $\mathbf{B}\left(s_t, \boldsymbol{\kappa}(s_t)\right)$. 

Without further restrictions the structural parameters in $\mathbf{B}$ cannot be recovered from the reduced form. To solve the identification problem, exclusion restrictions can be used \citep[e.g.,][impose various exclusion restrictions to identify the monetary policy shock]{leeperWhatDoesMonetary1996a,Primiceri2005,Sims2006,Wu2016}. These restrictions impose that some contemporaneous relations in $\mathbf{B}$ are zero based on theoretical arguments and are commonly assumed to hold for the whole sample. However, which variables are contemporaneously related might change over time, and, hence, how the  structural shocks are identified could be time-varying. 

Our model allows for such time-variation in the identifying exclusion restrictions by modelling the relationship between the reduced- and structural-form shocks in equation~\eqref{eq:sf} regime-dependent. First, the parameter estimates in our structural matrix $\mathbf{B}\left(s_t, \boldsymbol{\kappa}(s_t)\right)$ depend on a regime indicator $s_t$ of a discrete Markov process \citep[see][]{hamilton1989} with $M$ states. We assume that the MS process, $s_t$, is stationary, aperiodic, and irreducible, with a $M\times M$ transition matrix $\mathbf{P}$ and an $N$-vector of initial values $\boldsymbol{\pi}_0$, denoted by $s_t \sim\mathcal{M}arkov(\mathbf{P}, \boldsymbol{\pi}_0)$. 

Second and new to the literature, exclusion restrictions on the contemporaneous relations to identify structural shocks can vary over time, denoted by the dependence on a regime-specific collection of TVI indicators, $\boldsymbol{\kappa}(s_t)$. The TVI indicator selects amongst $K$ patterns of the exclusion restrictions imposed on the $n^{\text{th}}$ equation for each regime. Thus, we can select the $n^{\text{th}}$ structural shocks' identification patterns which are specific to the regimes identified by the MS process. Hence, we do not only allow for time-variation in the exclusion restrictions but we endogenously determine which restrictions are supported over time. We do not associate \textit{a priori} an identifying pattern with a specific regime but estimate for each of the $K$ exclusion restriction schemes the occurrence probability within each regime. In effect, a MS model without TVI might be estimated and the TVI will only occur if data supports it. Moreover, irrespective of TVI, MS in the structural coefficients enables time-varying impulse responses, which  \cite{LutkepohSchlaak2022} find to be important in heteroskedastic models.  

While in general the exclusion restrictions itself are sufficient to achieve partial identification of the $n^{\text{th}}$ shock -- assuming that the researcher sets them according to the identification results in \cite{rrwz2010} -- just imposing them would be not enough to endogenously select which are supported over time. In fact, the data-driven selection among the exclusion restriction patterns is possible if these are overidenifying. We achieve this by identification of the $n^{\text{th}}$ shock through heteroskedaticity within a regime. To that end, our model includes heteroskedastic structural shocks that follow SV. We specify that the structural shocks at time $t$ are contemporaneously and temporarily uncorrelated and jointly conditionally normally distributed given the past observations on vector $\mathbf{y}_t$, denoted by $\mathbf{Y}_{t-1}$, with zero mean and a diagonal covariance matrix:
\begin{align}
\mathbf{u}_t \mid \mathbf{Y}_{t-1} &\sim\mathcal{N}_N\left(\mathbf{0}_N, \diag\left(\boldsymbol{\sigma}_t^2\right)\right) \label{eq:sfshock}
\end{align}
where $\boldsymbol{\sigma}_t^2$ is an $N$-vector of structural shocks' conditional variances at time $t$ filling in the main diagonal of the covariance matrix, and $\mathbf{0}_N$ is a vector of $N$ zeros. 

Notably, our model now uses two different specifications to model time-variation: we model regime-dependence of structural parameters and exclusion restrictions with an MS process while we model heteroskedasticity with SV. This way, we can achieve identification through heteroskedasticity within a regime which we discuss in detail in section \ref{sec:identheteroskedasticity}. Broadly speaking, a percentage of the observations is allocated to a specific regime. This subset of observations can exhibit heterosekcastic residuals within the regime as the residuals follow a SV process allowing for different conditional variances at each point in time. Importantly, our specification implies that the MS process is driven by changes in the identification pattern and the values of structural parameters. The former are extreme as either a zero is imposed or not. Hence, only overwhelming support of the data for an exclusion restriction would lead to a sharp probability to be in a specific regime. 

\subsection{Time-varying identification}\label{sec:TVI}

\noindent To search for changes in the identification pattern of the $n^{\text{th}}$ shock structural shock across regimes, we develop the TVI mechanism and implement it through a new hierarchical prior distribution.  We set a multinomial prior distribution for the TVI component indicator $\boldsymbol{\kappa}(m) = k \in\{1,\dots,K\}$ with flat probabilities equal to $\frac{1}{K}$, denoted by 
\begin{align} \label{eq:priormultinom}
\boldsymbol{\kappa}(m) \sim\mathcal{M}ultinomial\left(K^{-1},\dots,K^{-1}\right).
\end{align} The indicator values, $\boldsymbol{\kappa}(m) = k \in\{1,\dots,K\}$, are associated with a specific restriction pattern. Hence, this prior setup allows us to estimate posterior probabilities for each TVI component in each regime and provides data-driven evidence in favour of the underlying identification pattern. The indiscriminate multinomial prior reflects our agnostic view of which identification pattern applies.

We combine the prior in equation \eqref{eq:priormultinom} with a 3-level equation-specific local-global hierarchical prior on the structural matrix.  The prior guarantees high flexibility and avoids arbitrary choices as the level of shrinkage is estimated within the model. For simplicity, denote by $\mathbf{B}_{m.k}$ the matrix $\mathbf{B}\left(s_t, \boldsymbol{\kappa}(s_t)\right)$ for given realisations of the MS process $s_t=m$, where $m$ denotes one of the $M$ regimes, and the TVI indicator $\boldsymbol{\kappa}(m)=k$, where $k$ stands for one out of $K$ exclusion restrictions pattern. To implement different restriction patterns, we follow \cite{WaggonerZha2003} and decompose the $n^{\text{th}}$ row of the structural matrix, $[\mathbf{B}_{m.k}]_{n\cdot}$, into a $1\times r_{n.m.k}$ vector $\mathbf{b}_{n.m.k}$, collecting the unrestricted elements to be estimated, and an $r_{n.m.k}\times N$ matrix $\mathbf{V}_{n.m.k}$, containing zeros and ones placing the elements of $\mathbf{b}_{n.m.k}$ at the appropriate spots, $[\mathbf{B}_{m.k}]_{n\cdot} = \mathbf{b}_{n.m.k} \mathbf{V}_{n.m.k}$,
where $r_{n.m.k}$ is the number of parameters to be estimated for the restriction pattern $k$. The restriction matrix $\mathbf{V}_{n.m.k}$ reveals which parameters are estimated and on which an exclusion restriction is imposed for regime $m$ and restriction pattern $k$. We have $K$ such restriction matrices for the $n^{\text{th}}$ row and the $m^{\text{th}}$ regime, each associated with a TVI component indicator $\boldsymbol{\kappa}(m) = k \in\{1,\dots,K\}$.

Given the fixed regime $s_t = m$ and sampled TVI component indicator, $\boldsymbol{\kappa}(m)=k$, we assume that  $\mathbf{b}_{n.m.k}$ is zero-mean normally distributed with an estimated shrinkage. This equation-specific level of shrinkage for $\gamma_{B.n}$ follows an inverted gamma 2 distribution with scale $\underline{s}_{B.n}$ and shape $\underline{\nu}_B$. The former hyper-parameter has a local-global hierarchical gamma-inverse gamma 2 prior, with the global level of shrinkage determined by $\underline{s}_{\gamma_B}$:
\begin{align}
\mathbf{b}_{n.m.k}' &\mid \gamma_{B.n}, \boldsymbol{\kappa}(m)=k \sim\mathcal{N}_{r_{n.m.k}}\left(\mathbf{0}_{r_{n.m.k}}, \gamma_{B.n} \mathbf{I}_{r_{n.m.k}}\right),  \qquad \text{with}\label{eq:priorbb}\\
\gamma_{B.n}\mid \underline{s}_{B.n} &\sim\mathcal{IG}2\left( \underline{s}_{B.n}, \underline{\nu}_B \right), \quad \underline{s}_{B.n} \mid \underline{s}_{\gamma_B} \sim\mathcal{G}\left( \underline{s}_{\gamma_B}, \underline{\nu}_{\gamma_B} \right), \quad \underline{s}_{\gamma_B} \sim\mathcal{IG}2\left( \underline{s}_{s_B}, \underline{\nu}_{s_B} \right),\label{eq:priorb}
\end{align} 
where $\mathbf{I}_{r_{n.m.k}}$ is the identity matrix. The subscript $n$ denotes an equation-specific parameter. In our empirical application we set $\underline{\nu}_B=10, \underline{\nu}_{\gamma_B}=10 ,  \underline{s}_{s_B}=100,$ and $\underline{\nu}_{s_B}= 1$ to allow a wide range of values for the structural matrix elements and show extraordinary robustness of our results that applies as long as the shrinkage towards the zero prior mean is not too imposing. This conditional normal prior is equivalent to that by \cite{WaggonerZha2003} for a time-invariant model and to a variant of the generalized--normal prior of \cite{arias2018inference}.  

Our TVI prior specification is a multi-component generalisation of \cite{geweke1996variable}'s spike-and-slab prior. We show this in the following Corollary.

\begin{corollary} \label{cl2}
Let $[\mathbf{B}_{m}]_{n.i}$ denote the $(n,i)^{\text{th}}$ element of the regime-specific structural matrix $\mathbf{B}(m, \boldsymbol{\kappa}(m))$ for the $m^{\text{th}}$ regime with  indicator $\boldsymbol{\kappa}(m)=k$, where $k$ stands for one out of $K$ exclusion restrictions pattern. Consider the multinomial prior distributions for $\boldsymbol{\kappa}(m)=k$ as in \eqref{eq:priormultinom} and the 3-level equation-specific local-global hierarchical prior on $\mathbf{b}_{n.m.k}$ as in \eqref{eq:priorb}--\eqref{eq:priorbb}, where $\mathbf{b}_{n.m.k}$ collects the unrestricted elements of the $n^{\text{th}}$ row of $\mathbf{B}_{m.k}$. Let $K_{R}$ denote the components in which the element $[\mathbf{B}_{m}]_{n.i}$ is restricted to zero.

Then, the prior distribution for $[\mathbf{B}_{m}]_{n.i}$, marginalized over $\boldsymbol{\kappa}(m)$ is
\begin{align}
	[\mathbf{B}_{m}]_{n.i} \mid \gamma_{B.n} &\sim \frac{K - K_{R}}{K}\mathcal{N}(0, \gamma_B) + \frac{K_{R}}{K}\delta_0.
\end{align}
\end{corollary} 
\begin{proof}
The element $[\mathbf{B}_{m}]_{n.i}$ is restricted to zero in $K_{R}$ components, whereas it stays unrestricted and normally distributed in $K-K_{R}$ of them. Therefore, its prior distribution is a Dirac mass at value zero, $\delta_0$, with probability $\frac{K_{R}}{K}$, and normal with probability $\frac{K - K_{R}}{K}$. 
\end{proof}

We can easily calculate which set of exclusion restrictions is best supported by the data by estimating the posterior probability for each restriction pattern. Having obtained $S$ posterior draws using a Gibbs sampler, this probability is estimated by a fraction of posterior draws of the TVI component indicator, $\kappa(m)^{(s)}$, for which it takes a particular value $k$, for each of its values from 1 to $K$:
\begin{align}
	\widehat\Pr\left[\kappa(m)= k\mid \mathbf{Y}_T\right] = S^{-1}\sum_{s=1}^{S}\mathcal{I}(\kappa(m)^{(s)}=k).\label{eq:tviprob}
\end{align} 

Focusing on the identification of the $n\textsuperscript{th}$ structural shock, we implement the TVI mechanism for the $n\textsuperscript{th}$ structural equation and, thus, provide partial identification of the $n\textsuperscript{th}$ structural shock. In general, the TVI mechanism can also be applied to multiple rows in the structural matrix and, hence, can search for TVI of multiple structural shocks making it adaptable to a broad set of applications.  It is then straightforward to calculate the conditional posterior distribution of the $n\textsuperscript{th}$ equation specification given a particular specification of the $i\textsuperscript{th}$ row. Also, if a specific economic interpretation is given to a particular combination of exclusion restrictions for the structural matrix, one can estimate the joint posterior probability of such a specification by computing the fraction of the posterior draws for which this combination of TVI indicators holds.

\subsection{Within-regime identification through heteroskedasticity} \label{sec:identheteroskedasticity}

\noindent Our search of the identification pattern is feasible only if data can distinguish between the different identifying patterns. Our TVI mechanism in the $n^{\text{th}}$ row of the structural matrix requires that the $n^{\text{th}}$ structural shock is identified within each regime, absent the potential additional restrictions we are searching for. 
We partially identify the $n^{\text{th}}$ structural shock and the $n^{\text{th}}$ row of the structural matrix within a regime through heteroskedasticity. The partial identification through heteroskedasticity focuses on the identification of a specific shock and all parameters of the corresponding row of the structural matrix within a regime without the necessity of imposing additional restrictions.

Identification through heteroskedasticity within regimes is feasible because the volatility of the structural shocks evolves according to a SV process. The volatility can have different values at each point in time. Instead, the structural matrix follows a MS process where time-variation occurs via allocating a subset of observations to a specific regime. Hence, structural residuals can still be heteroskedastic within the regime (and not only across regimes). Once the identification of a shock through heteroskedasticity is verified in each of the regimes, any set of exclusion restrictions overidentify the shock, which enables the data to discriminate between the alternative restriction patterns.

We obtain partial identification of the $n^{\text{th}}$ row of the regime-specific structural matrix, denoted by $\mathbf{B}_m$, where the dependence on $k$ is neglected, within each regime  under the following corollary. 
\begin{corollary} \label{cl1}
Consider the SVAR model with the structural matrix following MS from equations \eqref{eq:rf}--\eqref{eq:sfshock}. Denote by $t_1^{(m)}, \dots, t_{T_M}^{(m)}$ the time subscripts of the $T_m$ observations allocated to the Markov-switching regime $s_t = m \in \{1,\dots,M\}$, where $\sum_{m=1}^{M}T_m = T$. Let $\boldsymbol{\sigma}_{n.m}^2 = (\sigma_{n.t_1^{(m)}}^2,\dots,\sigma_{n.t_{T_M}^{(m)}}^2)$ be the vector containing all conditional variances $\sigma_{n.t}^2$ of the $n^{\text{th}}$ row associated with the observations in the $m\textsuperscript{th}$ regime. Then the $n^{\text{th}}$  row of $\mathbf{B}_m$ is identified up to a sign if 
$\boldsymbol{\sigma}_{n.m}^2 \neq \boldsymbol{\sigma}_{i.m}^2$ for all $i \in \{1, \dots, N\}\setminus \{n\}$.
\end{corollary} 
\begin{proof}
The proof proceeds by the same matrix result as in Corollary 1 by \cite{LSUW2022} set for a heteroskedastic SVAR model without time-variation in the structural matrix. In line with Theorem 1 by \cite{LutkepohlWozniak2017}, the matrix result holds even if the inequality above is true for one period.
\end{proof}

According to Corollary~\ref{cl1}, the $n^{\text{th}}$ row of $\mathbf{B}_{m.\mathbf{k}}$ is identified through heteroskedasticity in the $m\textsuperscript{th}$ regime if changes in the conditional variances of the $n\textsuperscript{th}$ structural shocks evolve non-proportionally to those of other shocks within that regime. 
This result assures that the corresponding impulse responses, the $n^{\text{th}}$ column of $\mathbf{B}_{m.\mathbf{k}}^{-1}$, are identified as well. 
Global identification of $\mathbf{B}_m$ through heteroskedasticity is achieved when Corollary \ref{cl1} holds for all  $n \in \{1, \dots, N\}$. 

It is essential to verify identification through heteroskedasticity within each regime \citep[see the argument in][]{Montiel2022} and, to that end, we introduce as a new feature of our model that the SV process of the structural shocks has a regime-dependent standard deviation parameter. Each of the $N$ conditional variances, $\sigma_{n.t}^2$, follows a non-centered SV process with regime-dependent standard deviation parameter that decomposes the conditional variances into 
\begin{align} 
	\sigma_{n.t}^2 &= \exp\{\omega_n(s_t) h_{n.t}\}\quad\text{with} \label{eq:SV}\\
	 h_{n.t} &= \rho_n h_{n.t-1} + v_{n.t} \quad\text{and}\quad v_{n.t} \sim \mathcal{N}(0,1)
\end{align}
where $\omega_n(s_t)$ is a standard deviation parameter of the log-conditional variances $\log\sigma_{n.t}^2$. Equation \eqref{eq:SV} implies that both $h_{n.t}$ and $\omega_n(s_t)$ are identified up to a simultaneous sign change without affecting the conditional variance $\sigma_{n.t}^2$. The log-volatility of the $n^{\text{th}}$ structural shock at time $t$,  $h_{n.t}$, follows an autoregressive model with initial value $h_{n.0} = 0$, where $\rho_n$ is the autoregressive parameter and $v_{n.t}$ is a standard normal innovation. 
Note that the unconditional expected value of the conditional variances $\sigma_{n.t}^2$ depends on $\omega_n(s_t)$ and $\rho_n$. To fix it, we normalise the structural shocks' conditional variances around one by a prior that shrinks the model towards homoskedasticity which we explain in section \ref{sec:priorandposterior}. The specification extends the SV process by \cite{Kastner2014} 
by the regime-dependence in the parameter $\omega_n(s_t)$. Therefore, the Markov process $s_t$ drives the time variation in both $\mathbf{B}\left(s_t,  \boldsymbol{\kappa}(s_t)\right)$ and $\omega_n(s_t)$.  

We can verify the hypothesis of within-regime homoskedasticity of each of the shocks by checking the restriction \citep[see][for time-invariant heteroskedastic models]{Chan2018}
\begin{align}
\omega_{n}(m) = 0. \label{eq:omegarestric}
\end{align}
If it holds, the $n\textsuperscript{th}$ structural shock is homoskedastic in the $m\textsuperscript{th}$ regime. If it does not, then the structural shock's conditional variance changes non-proportionally to those of other shocks with probability 1, rendering the shock identified through heteroskedasticity. Consequently, two cases guarantee the identification of the $n^{\text{th}}$ structural shock. First, this shock, and no other, is homoskedastic within a specific regime. Hence, the condition in equation~\eqref{eq:omegarestric} only holds for the $n^{\text{th}}$ shock. Second, the $n^{\text{th}}$ structural shock is heteroskedastic in that regime implying that restriction~\eqref{eq:omegarestric} is not satisfied for the $n^{\text{th}}$ shock.

\subsection{Priors, Normalisation, and Posterior Sampler}\label{sec:priorandposterior}

\noindent In this section, we discuss the prior specification of the remaining parameters of the model and their estimation. We use the following notation: $\mathbf{e}_{n.N}$ is the $n\textsuperscript{th}$ column of the identity matrix of order $N$, $\boldsymbol{\imath}_n$ is a vector of $n$ ones, $\mathbf{p}$ is a vector of integers from 1 to $p$, and $\otimes$ is the Kronecker product.

Let the $N\times (Np+d)$ matrix  $\mathbf{A}=\begin{bmatrix}\mathbf{A}_1 &\dots & \mathbf{A}_p & \mathbf{A}_d\end{bmatrix}$ collect all autoregressive and constant parameters. Each row of matrix $\mathbf{A}$ is independent and follows conditional multivariate normal distribution with a 3-level global-local hierarchical prior on the equation-specific shrinkage parameters $\gamma_{A.n}$:
\begin{align}
[\mathbf{A}]_{n\cdot}' &\mid \gamma_{A.n} \sim\mathcal{N}_{Np+d}\left(\underline{\mathbf{m}}_{A.n}', \gamma_{A.n} \underline{\boldsymbol{\Omega}}_A\right),  \qquad \text{with}\\
\gamma_{A.n}\mid \underline{s}_{A.n} &\sim\mathcal{IG}2\left( \underline{s}_{A.n}, \underline{\nu}_A \right), \underline{s}_{A.n} \mid \underline{s}_{\gamma_{A.n}} \sim\mathcal{G}\left( \underline{s}_{\gamma_A}, \underline{\nu}_{\gamma_A} \right), \underline{s}_{\gamma_A} \sim\mathcal{IG}2\left( \underline{s}_{s_A}, \underline{\nu}_{s_A} \right),
\end{align}
where $\underline{\mathbf{m}}_{A.n} = \begin{bmatrix}\mathbf{e}_{n.N}' & \mathbf{0}_{N(p-1)+d}' \end{bmatrix}$,  $\underline{\boldsymbol{\Omega}}_A$ is a diagonal matrix with vector $\begin{bmatrix}\mathbf{p}^{-2\prime}\otimes\boldsymbol{\imath}_N' & 100\boldsymbol{\imath}_d'\end{bmatrix}'$ on the main diagonal, hence, incorporating the ideas of the Minnesota prior of \cite{Doan1984}. In our empirical application we set $\underline{\nu}_A, \underline{\nu}_{\gamma_A}, \underline{s}_{s_A},$ and $\underline{\nu}_{s_A}$ all equal to 10, which facilitates relatively strong shrinkage that gets updated, nevertheless. Providing sufficient flexibility on this 3-level hierarchical prior distribution occurred essential for a robust shape of the estimated impulse responses.

Each row of the transition probabilities matrix $\mathbf{P}$ of the Markov process follows independently a Dirichlet distribution, $[\mathbf{P}]_{m\cdot} \sim\mathcal{D}irichlet(\boldsymbol{\imath}_M + d_m \mathbf{e}_{m.M})$, as does the initial regime probabilities vector $\boldsymbol{\pi}_0$, $\boldsymbol{\pi}_0 \sim\mathcal{D}irichlet(\boldsymbol{\imath}_M)$.
To assure the prior expected regime duration of 12 months, we set $d_m = 11$ for models with $M=2$ in our application, and show the robustness of our results to alternative prior expected regime durations.

The prior distribution for the regime-specific standard deviation parameter, $\omega_n(s_t)$, follows a conditional normal distribution, $	\omega_n(m) \mid \sigma^2_{\omega.n} \sim\mathcal{N}\left( 0, \sigma^2_{\omega.n}\right)$, with a gamma distribution on the shrinkage level, $\sigma^2_{\omega.n} \sim\mathcal{G}(1, 1)$.  
Notably, the marginal prior for $\omega_n(m)$ combines extreme prior probability mass around the restriction for homoskedasticity $\omega_n(m)=0$ and fat tails. The latter enables efficient extraction of the heteroskedasticity signal from data assuring that a shock is heteroskedastic within a regime only if the data favour this outcome.  Our empirical results remain unchanged subject to variation of the prior scale of $\sigma^2_{\omega.n}$ ranging from values 0.1 to 2. Finally, the prior distribution for the autoregressive parameter of the log-volatility process is uniform over the stationarity region, $\rho_n \sim\mathcal{U}(-1,1)$. 

This prior setup following \cite{LSUW2022} implies that the structural shocks' conditional variance, $\sigma^2_{n.t}$, have a log-normal-product marginal prior implied by the normal prior distributions for the latent volatility, $h_{n.t}$, and for the standard deviation parameters, $\omega_{n}(m)$. The log-normal-product marginal prior has a pole at value one, corresponding to homoskedasticity, fat tails, and a limit at value zero when the conditional variance goes to zero from the right. Importantly, the extreme concentration of the prior mass at one and strong shrinkage toward this value standardises the structural shocks' conditional variances around value $\sigma_{n.t}^2 = 1$. This assures appropriate scaling in the estimation of the structural matrix. In our empirical model, the normalisation holds true as the estimated variances fluctuate closely around one.

To draw from the posterior distributions we use a Gibbs sampler. We provide an \textbf{R} package \textbf{bsvarTVPs} facilitating the reproduction of our results as well as new applications.\footnote{Access the \textbf{bsvarTVPs} package at \url{https://github.com/bsvars/bsvarTVPs}}

\section{Time-varying identification of US monetary policy shocks} \label{sec:resultsTVI}

\noindent Over the last sixty years, economic conditions, the arrival of new technologies, understanding the role of communication, and new strategies of conducting monetary policy in the US seem to affect the most efficient ways of extracting the structural shocks from data. Despite these circumstances, empirical studies on the effects of the US monetary policy shock assume that restrictions used to identify the structural shocks do not change over time.

We challenge this view and study time-variation in US monetary policy shock identification using data from January 1960 to September 2024 on log differenced industrial production, $y_t$, log consumer price index inflation, $\pi_t$, federal funds rates, $R_t$, term spread, $TS_t$, measured as the 10-year treasury constant maturity rate minus the federal funds rate, log M2 money supply, $m_t$,  and the log S\&P500 stock price index, $sp_t$ in a model with six lags. The Online Appendix gives details on the data. 

\begin{table}[t] 
	\caption{Monetary policy shocks identification patterns }\label{tab:Bidentrest} \vspace{-0.5cm}
 \begin{center}
	\begin{tabular*}{\textwidth}{@{}lccccccclcccccc@{}} 
	& \multicolumn{6}{c}{third row} & & & \multicolumn{6}{c}{fourth row}\\ \cline{1-7} \cline{9-15}
	& $y_t$ & $\pi_t$ & $R_t$ & $TS_t$ & $m_t$ &  $sp_t$ && & $y_t$ & $\pi_t$ & $R_t$ & $TS_t$ & $m_t$ &  $sp_t$\\\cline{2-7} \cline{10-15}	

	baseline &$*$ &$*$ & $*$ & 0 & 0 & 0 && unrest. &$*$ &$*$ & $*$ & $*$  & $*$  & $*$   \\\cline{2-7}\cline{10-15}
	with $TS$ &$*$ &$*$ & $*$ & $*$ & 0 & 0  && with 0 &$*$ &$*$ & 0 & $*$  & $*$  & $*$\\\cline{2-7}\cline{10-15}
	with $m$ &$*$ &$*$ & $*$ & 0 &$*$ & 0  && & & & & & &\\\cline{2-7}	
	 only $m$ &0 &0 & $*$ & 0 &$*$ & 0  && & & & & & &\\\cline{1-7} \cline{9-15}
	\end{tabular*}
\end{center}
\footnotesize
\renewcommand{\baselineskip}{11pt}
\textbf{Note:} $*$ indicates an unrestricted parameter and $0$ an exclusion restriction on the contemporaneous parameters in the third or fourth row of $\mathbf{B}$.

\end{table}

We specify four ($K_3=4$) exclusion restriction patterns for the third row to identify the US monetary policy shock, shown in the left panel of Table~\ref{tab:Bidentrest}.\footnote{We found that an unrestricted third row does not provide economically interpretable monetary policy shocks. On the contrary, the additional exclusion restrictions sharpen the identification and interpretations.  } The first identification pattern, that we call \textit{baseline}, imposes simultaneity in short-term interest rates, prices and output  \citep{leeperWhatDoesMonetary1996a}. The second identification pattern, \textit{with} $TS$, relatively to the baseline, allows additional simultaneous movements of the term spread, $TS_t$, in spirit of \cite{Vazquez2013} and \cite{diebold2006}. We consider this identification in response to the monetary policy shock identification strategies by \cite{Baumeister2013}, \cite{Feldkircher2016}, and \cite{Liu2017} but emphasise the spread's role in determining the interest rates. 

The third identification, \textit{with} $m$, allows for simultaneity in interest rates and money, on top of output and prices, inspired by \cite{Andres2006} and \cite{BelongiaIreland2015}. This identification suggests synchronous movements in interest rates and market liquidity levels. The fourth identification pattern, \textit{only} $m$, for the monetary policy shock highlights the simultaneity between short-term interest rates and a monetary aggregate, in line with the identification used by \cite{LeeperRoush2003} and  \cite{Sims2006}. 

We specify two ($K_4=2$) restriction patterns for the fourth row, shown in the right panel of Table~\ref{tab:Bidentrest}. Here, we either set no restrictions, labeled \textit{unrestricted}, or impose a zero on the simultaneous movement of term spread with interest rates, labeled \textit{with 0}. This zero restriction implies that  monetary policy interventions can have a narrowing effect on the term spread while the short term interest rate does not move. Such a restriction on the contemporaneous relations is in spirit with the notion that short term interest rates do not react to a term spread shock for four quarters imposed by \cite{Baumeister2013} and \cite{Feldkircher2016}. We impose a lower-triangular structure on the remaining rows of the structural matrix. 

\subsection{Model evaluation}\label{ssec:modeleval}

In order to verify whether data support TVI, we evaluate the forecasting performance of models.\footnote{Forecasting performance measures are unbiased as they are invariant to local identification of the SVARs up to the signs and ordering of structural equations, the signs of the SV log-volatility components, and MS regime labeling \citep[see][]{Geweke2007}.} More specifically, we extend a heteroskedastic SVAR in two dimensions: by applying a different number of regimes $M\in\{1,2,3\}$, and by using models with the multinomial prior over the identification patterns listed in Table~\ref{tab:Bidentrest} as well as models with fixed identification patterns over MS regimes.

We use relative predictive log-scores, as defined by \cite{geweke2010comparing}, to evaluate density forecasts and relative root-mean-squared forecast errors for point forecasts. In both cases the benchmark model is the heteroskedastic SVAR with $M=2$ MS regimes and the TVI selection mechanism proposed in this paper. Therefore, a relative log-score higher than 0 or a relative root-mean-squared forecast error lower than 1 indicates that a model outperforms the benchmark. We compute the measures based on 190 one- and twelve-step-ahead forecasts with the forecast origins starting in December 2007 and ending in September 2023.

\begin{figure}[t]
\subfloat[Density forecast performance]{	\includegraphics[trim={0.0cm 0cm 0.0cm 1cm},width=0.5\textwidth]{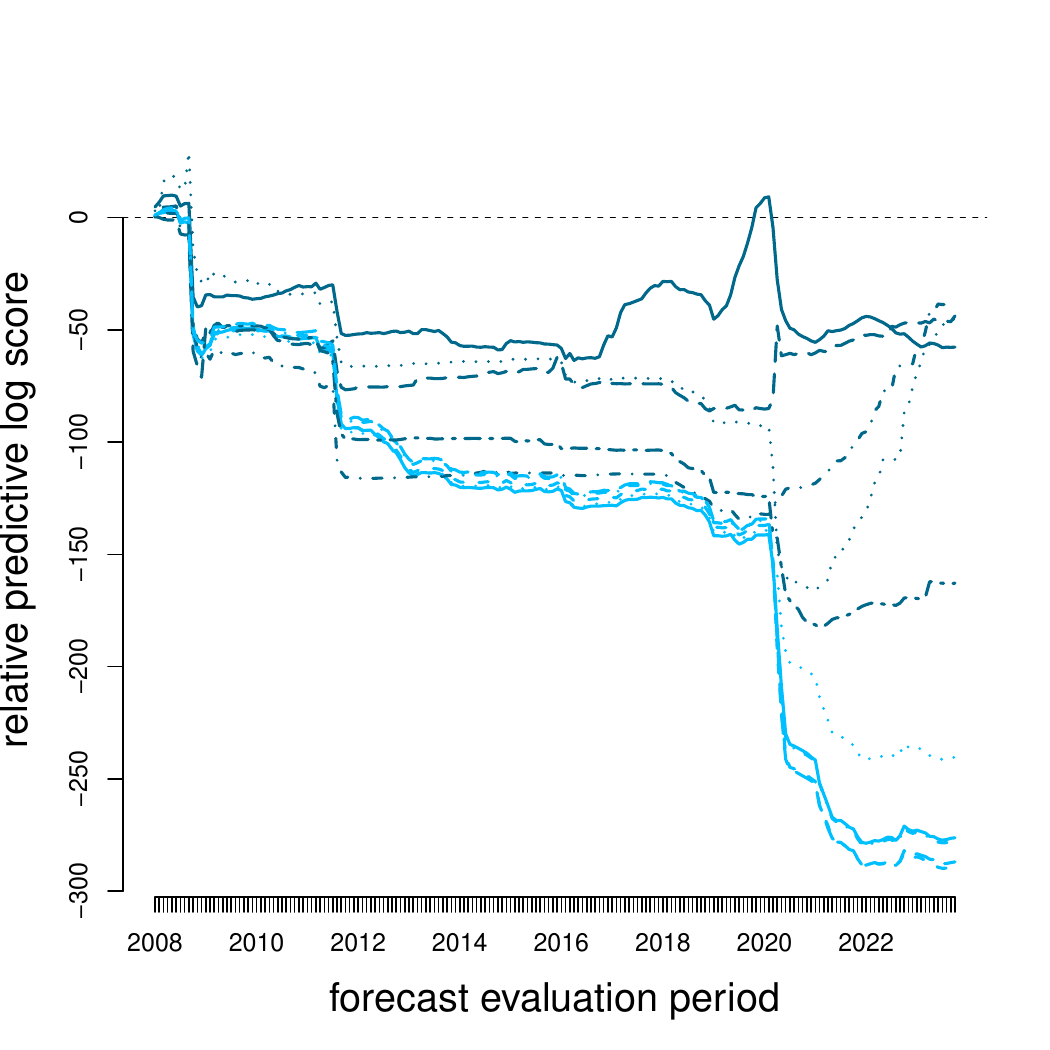}} 
\subfloat[Point forecast performance]{	\includegraphics[trim={0.0cm 0cm 0.0cm 1cm},width=0.5\textwidth]{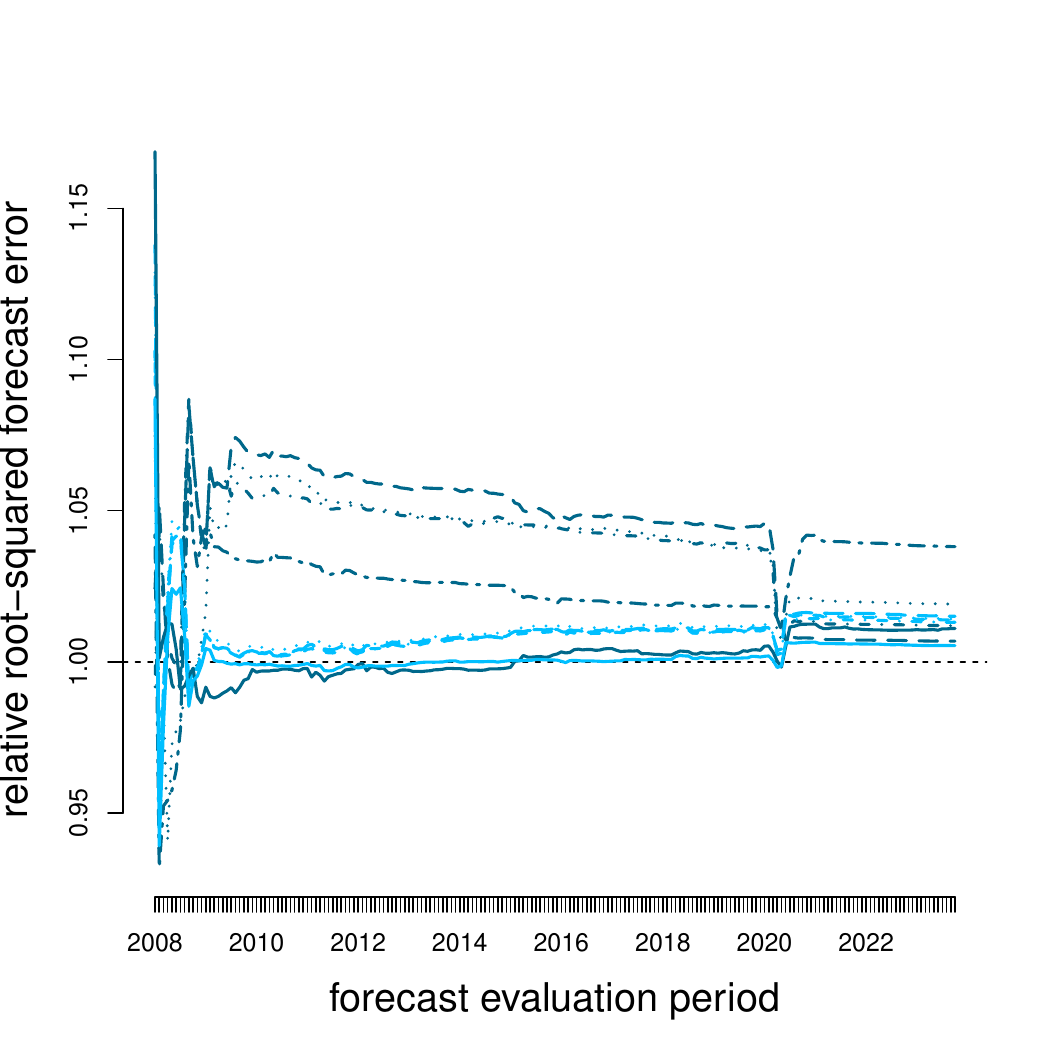}} 
\caption{Forecasting performance of selected models } \label{fig:forecasting}
\footnotesize
\renewcommand{\baselineskip}{11pt}
\textbf{ Note: }The legend for models includes: MS models (dark colour): \includegraphics[scale=0.09]{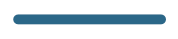} $M=3$ with TVI, \includegraphics[scale=0.09]{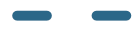}~$M=2$ baseline, \includegraphics[scale=0.09]{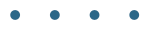} $M=2$ with $TS$, 	\includegraphics[scale=0.09]{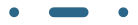} $M=2$ with $m$, \includegraphics[scale=0.09]{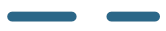} $M=2$ only $m$, no MS ($M=1$, bright colour): \includegraphics[scale=0.09]{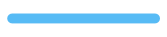} with TVI, \includegraphics[scale=0.09]{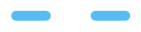} baseline, \includegraphics[scale=0.09]{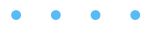} with $TS$, \includegraphics[scale=0.09]{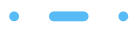} with $m$, \includegraphics[scale=0.09]{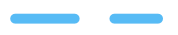} only $m$. Measures are relative to the benchmark model with $M=2$ and TVI.
\end{figure}

Figure \ref{fig:forecasting} reports the one-step-ahead forecasting performance based on density (panel a) and point forecasts (panel b). We find strong support for using TVI and MS. Overall, the model with two regimes and TVI used as the benchmark preforms best according to both measures, followed by the model with three regimes and TVI (solid dark line). The latter outperforms the benchmark at the beginning of the sample and around 2020 based on density forecasts and roughly until 2015 based on point forecasts. The model with TVI but no MS (solid bright line) exhibits accurate point forecasting performance, reasonably close to the benchmark, but performs poorly 
in density forecasting. Models with TVI outperform the alternatives with no TVI (dashed or dotted lines). The twelve-steps-ahead forecasting performance, discussed in the Online Appendix, provides a bit more nuanced picture. Still, the benchmark model is not uniformly outperformed by any other model according to both measures.

Based on these findings, we provide more insights using the benchmark model with two MS regimes and TVI.

\subsection{TVI probabilities}\label{ssec:TVIprob}

\noindent We find that data support TVI of US monetary policy shocks.  To show this, first, we estimate the posterior probability of the structural matrix featuring different identification in the two regimes:
$\Pr\left[\kappa_3(m=1)\neq \kappa_3(m=2) \lor \kappa_4(m=1)\neq \kappa_4(m=2) \mid \mathbf{Y}_T\right]$.
Given our model setup, the probability includes the intersection of changing the identification in both rows, third and fourth. In the two estimated MS regimes, data favor different identification schemes with strong posterior support of 0.87. 

Moving on to the monetary policy shock identification, we find remarkably sharp posterior probabilities of the TVI indicators reported in Figure~\ref{fig:TVI}. In the first regime indicated by the light color in Figure~\ref{fig:TVI}, we find strong data support that the monetary policy shock is identified by simultaneous movements of interest rates with output, inflation, and the term spread, combined with no restrictions in the forth row on the simultaneous movements of term spreads to all other variables. The posterior probability of the TVI component indicators for the restriction pattern \textit{with} $TS$ and \textit{unrestricted} are numerically equal to one. This suggests that the monetary policy shock is best identified with consideration to the whole bond yield curve signaling adjusted expectations on future economic developments in the spirit of \cite{diebold2006}. 

\begin{figure}[t]
	\includegraphics[trim={0.0cm 1.0cm 0.0cm 2.0cm},width=1\textwidth]{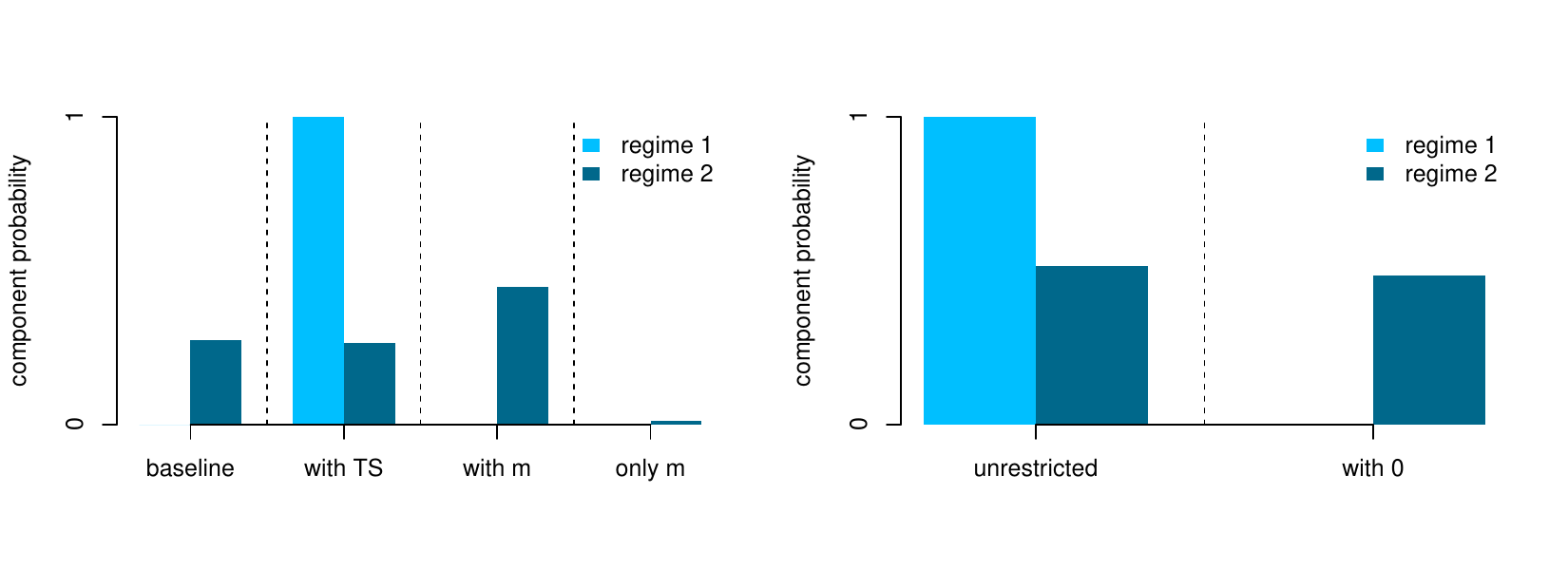} 
		\begin{center}
\begin{tabular}{p{1in}p{2in}p{2in}}
&(a) \textit{third row} & (b) \textit{fourth row}
\end{tabular}
\end{center}
	\caption{Posterior probabilities of the TVI components in both regimes} \label{fig:TVI} 
\end{figure}

In the second regime visualized by the dark color in Figure~\ref{fig:TVI}, we find evidence that data support a mixed identification of the monetary policy shock. The highest posterior probability is attached to the restriction pattern allowing interest rates also to respond simultaneously to changes in the monetary aggregate. The posterior value of the TVI component indicator for the third exclusion restriction pattern (\textit{with} $m$) is 0.45. The posterior probability of the remaining components is lower, with first and second both having 0.27 posterior probability and the fourth (\textit{only} $m$) 0.01. 

Our findings are in line with \cite{BelongiaIreland2015} supporting that after 2008, 
the monetary policy shock is best identified considering the high variability of the monetary aggregates. While the Fed officially targeted money supply before 1982, \cite{BelongiaIreland2015} argue that unconventional monetary policy can be seen as an attempt to increase money growth. Similar to the support of data for money in the identification pattern for the monetary policy shock in the second regime, the authors show that excluding a monetary measure from the interest rate rule is rejected by the data in the sample to 2007. Similarly, \cite{LutkepohlWozniak2017} find empirical evidence for monetary policy shock identification allowing for simultaneity in interest rates and money for data up to 2013. Neither of these papers considers identification pattern including the term spread or time variation of the structural matrix.

The posterior probabilities of the TVI component indicator for the fourth row show a dominant role of the unrestricted identification in the first regime, and equal support for \textit{unrestricted} and \textit{with 0} in the second one. This indicates that in the second regime, which predominately covers the period of the zero-lower-bound environment, monetary policy interventions partly have a narrowing effect on the term spread while the short term interest rate does not move. However, since the second regime also covers periods outside the zero lower bound period, we find only weak empirical support for the restriction which is in spirit of the assumption imposed by \cite{Baumeister2013} and \cite{Feldkircher2016} that short term interest rates do not immediately react to a term spread shock. 

Selecting two different identification patterns over time is remarkably robust with respect to extending or limiting the set of possible exclusion restrictions in the search (such as allowing for interest rates to move simultaneously with stock prices, including patterns with zero restrictions on money and stock prices or zero restrictions on interest rates, money and stock prices in the fourth row, excluding the variant \textit{only m}, or allowing only for simultaneous movements of interest rates, output, and inflation with term spread or money in the third row). Moreover, models with changes in the prior set-ups on the contemporaneous relations, or changes in the variables (using the consumer price index and industrial production index instead of growth rates, or the interpolated GDP deflator instead of CPI inflation) support time-variation in the identification of the monetary policy shock.\footnote{While we focus on reporting robustness to the changes in the prior specification most relevant to the discussed feature, all our main findings regarding TVI (Section~\ref{ssec:TVIprob}), identification through heteroskedasticity (\ref{ssec:IVH}), and Markov-switching probabilities (\ref{ssec:MS}) are robust to changing hyper-parameters on the priors set on the contemporaneous relationships, autoregressive coefficients, SV or the Markov-process, to different lag lengths (from 1 to 12), and using log first differences of money and stock prices. Results are available upon request.} We report the robustness results in the Online Appendix.

\subsection{Within-regime identification through heteroskedasticity}\label{ssec:IVH}

\noindent Our search for regime-specific identification patterns requires that the exclusion restrictions are over-identifying and, thus, can be tested using statistical methods. We achieve this by identifying the monetary policy shock through heteroskedasticity. 

\begin{figure}[h]
\subfloat[$\omega_{3}$]{	\includegraphics[trim={0.0cm 0.5cm 0.0cm 2.0cm},width=0.5\textwidth]{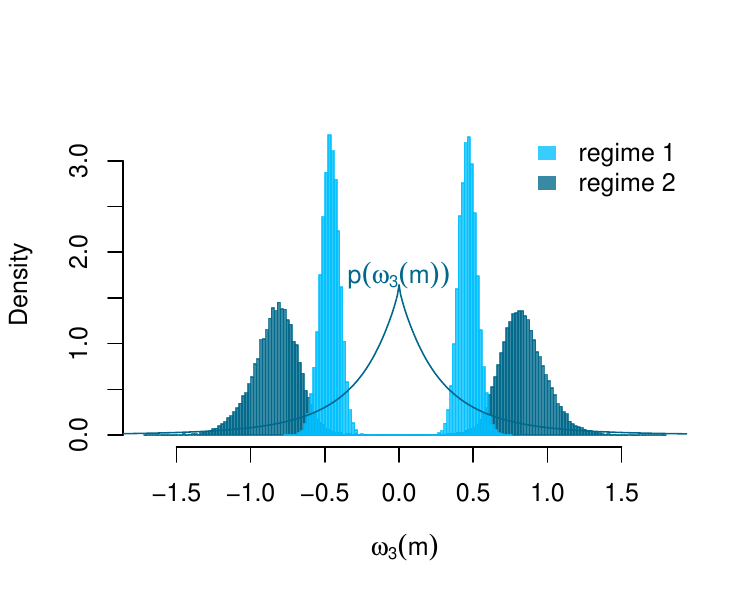}} 
\subfloat[$\omega_{4}$]{	\includegraphics[trim={0.0cm 0.5cm 0.0cm 2.0cm},width=0.5\textwidth]{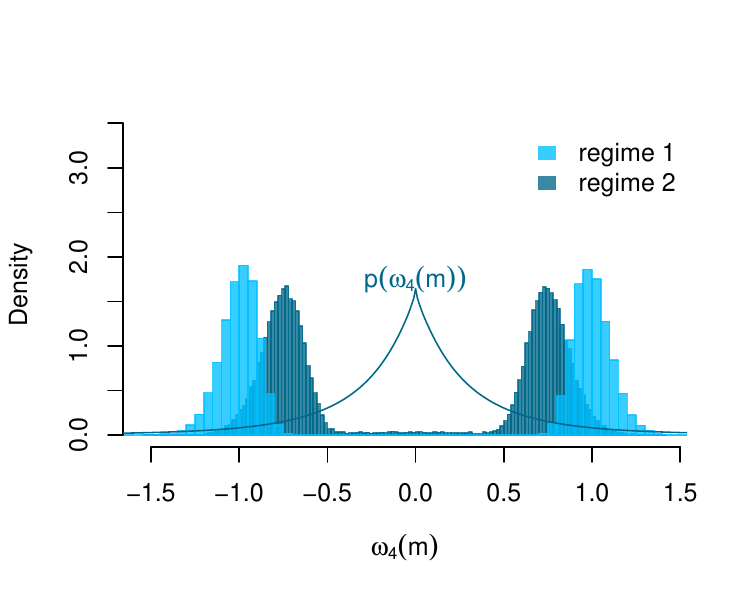}} 
	\caption{Marginal posterior and prior densities of the regime-specific standard deviations } \label{fig:omega}
\end{figure}

Figure~\ref{fig:omega} plots histograms of the standard deviation of the third shock $\omega_3(m)$ and of the fourth shock $\omega_4(m)$ for both regimes, $m\in\{1,2\}$. We find clear evidence for heteroskedasticity within both regimes for the third and fourth equation visualized by the bi-modality of the posterior distributions. The bi-modality results from the fact that $\omega_3(m)$ and $\omega_4(m)$, as well as $h_{n.t}$, are identified up to their simultaneous change of sign, which is implied by equation \eqref{eq:SV}. This is taken into account when verifying heteroskedasticity via evaluating  $\omega_n(m)=0$ following \cite{Chan2018}. Verifying that $\omega_3(m)$ and $\omega_4(m)$ differ from zero implies identifying the underlying equations. In both cases, the mass of the posterior distributions is away from zero. Therefore, the third and fourth row are identified through heteroskedasticity and our selection of identifying restrictions is valid in both regimes. We then label the third shock as the monetary policy shock based on the exclusion restrictions imposed. 

We find that the identification of the third and fourth row through heteroskedasticity is robust to changing the scale hyper-parameter of the prior controlling for the regime-dependent standard deviation parameter level to 0.5 or 2, see results in the Online Appendix. Also, using alternative measurements for our variables, such as the consumer price index, industrial production, or the interpolated GDP deflator, has no impact on identification through heteroskedasticity.

\subsection{Two Markov-switching regimes}\label{ssec:MS}

\noindent Having established that the data support regime-specific identification of US monetary policy shocks, we now take a closer look at the two estimated MS regimes. Figure~\ref{fig:regimeprob} shows the estimated regime probabilities for the second regime (a gray shaded area), together with highlighted periods when the probability for the second regime is larger than 0.8 (dark-colored thick line), and specific events (black vertical lines topped by letters).

\begin{figure}[t!]
	
		\includegraphics[trim={1.5cm 0.5cm 1.0cm 2.0cm},width=1\textwidth]{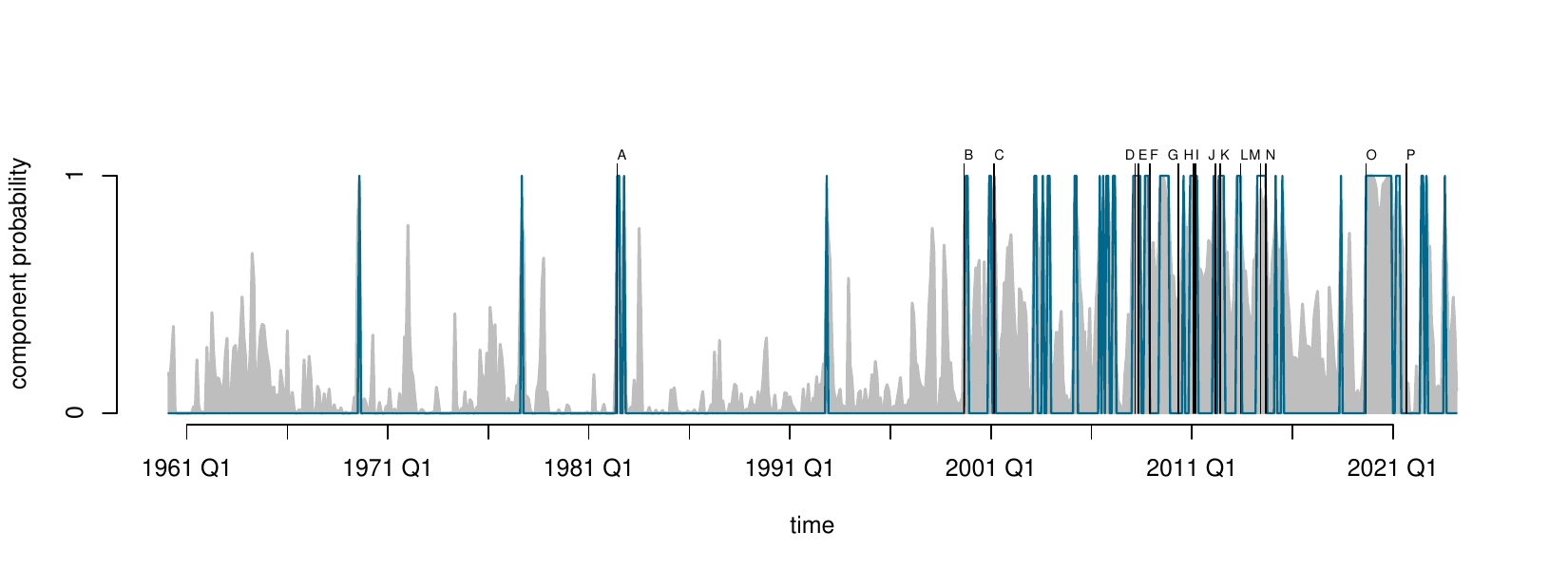}
		 \scriptsize
\textbf{List of events:}\\[1ex] 
	\begin{tabular}{p{0.2cm}p{1.3cm}p{11cm}} \hline\hline
		A &  Dec 1982& peak of 1981-1982 recession\\
		B	& Mar 2000 & dot-com bubble burst \\
		C	& Sep 2001  & September 11 attacks\\
		D & Sep 2008 & bankruptcy of Lehman Brothers on September 15, 2008 \\
		E & Nov 2008  & FOMC announces Quantitative Easing 1\\
		F & Jun 2009 & NBER declares end of the US great recession\\
		G	&		Nov 2010 & FOMC announces Quantitative Easing 2 \\
		H	&		Aug 2011 & FOMC announces to keep the federal funds rate  at its effective zero lower bound ``at least through mid-2013''$^*$\\ 
		I & Sep 2011  & FOMC announces Operation Twist\\
		J	&			Sept 2012 & FOMC announces Quantitative Easing 3 \\
		K	&		Dec 2012 & FOMC announces to purchase 45B\$ of longer-term Treasuries per month for the future, and to keep the federal funds rate at its effective zero lower bound\\\
		L	&	Dec 2013 & FOMC announces to start lowering the purchases of longer-term Treasuries and mortgage-backed securities (to \$40B and \$35B per month, respectively)\\
		M	&		Dec 2014 & FOMC announces  ``it can be patient in beginning to normalize the stance of monetary policy''\\
		N &		Mar 2015 & FOMC announces ``an increase in the target range for the federal funds rate remains unlikely at the April FOMC meeting''\\
		O & Mar 2020& start of COVID-19 pandemic. First unscheduled FOMC meeting on March 3 in response to the outbreak of the pandemic announcing to ``lower the target range for the federal funds rate by 1/2 percentage point'' and to ``use its tools and act as appropriate to support the economy.'' On March 23, FOMC statement to ``continue to purchase Treasury securities and agency mortgage-backed securities in the amounts needed to support smooth market functioning''.\\
		P & Mar 2022& First indications of changes in monetary policy with FOMC annoucing that one member voted against keeping the federal funds rate at the same rate. In May and June the federal funds rate was increased with the largest hikes since May 2000 and 1994, respectively. \\
		\hline\hline
\end{tabular}

 $^*$ all quotes from \url{https://www.federalreserve.gov/newsevents/pressreleases.htm}.
\caption{Regime probabilities of the Markov process} \label{fig:regimeprob}
\end{figure}

Regime probabilities are volatile and often not sharp. This is different from ``traditional'' MS-VAR models where the MS process is driven by changes in autoregressive parameters or volatility of the residuals. In our model, the MS probabilities predominantly capture changes in the contemporaneous structural relations.  Importantly, the two regimes are characterized by two different identification schemes as revealed by the estimated TVI probabilities. Hence, the estimated regime-specific structural matrices differ sharply by imposing a zero on a specific parameter versus estimating this parameter freely.   

We interpret high regime probabilities as overwhelming evidence that monetary shock identification is best captured by the TVI selected identification pattern during this period. Especially, the second regime's supported identification of a monetary policy shock requires strong simultaneous movements in money and interest rates. Hence, high probabilities for the second regime show up in times when data feature such strong simultaneity, which we discuss further in Section~\ref{ssec:3results}. The identification scheme selected in the second regime captures such data characteristics for the sake of sharper identification of the shocks. 

Regime one is predominant in the first part of the sample until 2000. It is present in periods of rather normal economic developments (aka non-crisis times). After 2000 the second regime occurs more frequent and the periods are characterized by frequent switches between the two regimes. The second regime clearly prevails during crisis periods, in particular the financial and COVID-19 ones. In the aftermaths of these events, regime 2 gains persistence. Additionally, it is present at earlier extra-ordinary times such as the peak of the 1981--1982 crisis, the burst of the dot-com bubble in March 2000, or 9/11. It also predominately covers periods of unconventional monetary policy actions such as Quantitative Easing 1, Operation Twist, and several subsequent announcements of the Fed, starting in August 2011, stating that they would keep the federal funds rate at its effective zero-lower bound for a substantial period.

\begin{figure}[t]
		\includegraphics[trim={1.5cm 0.5cm 1.0cm 2.0cm},width=1\textwidth]{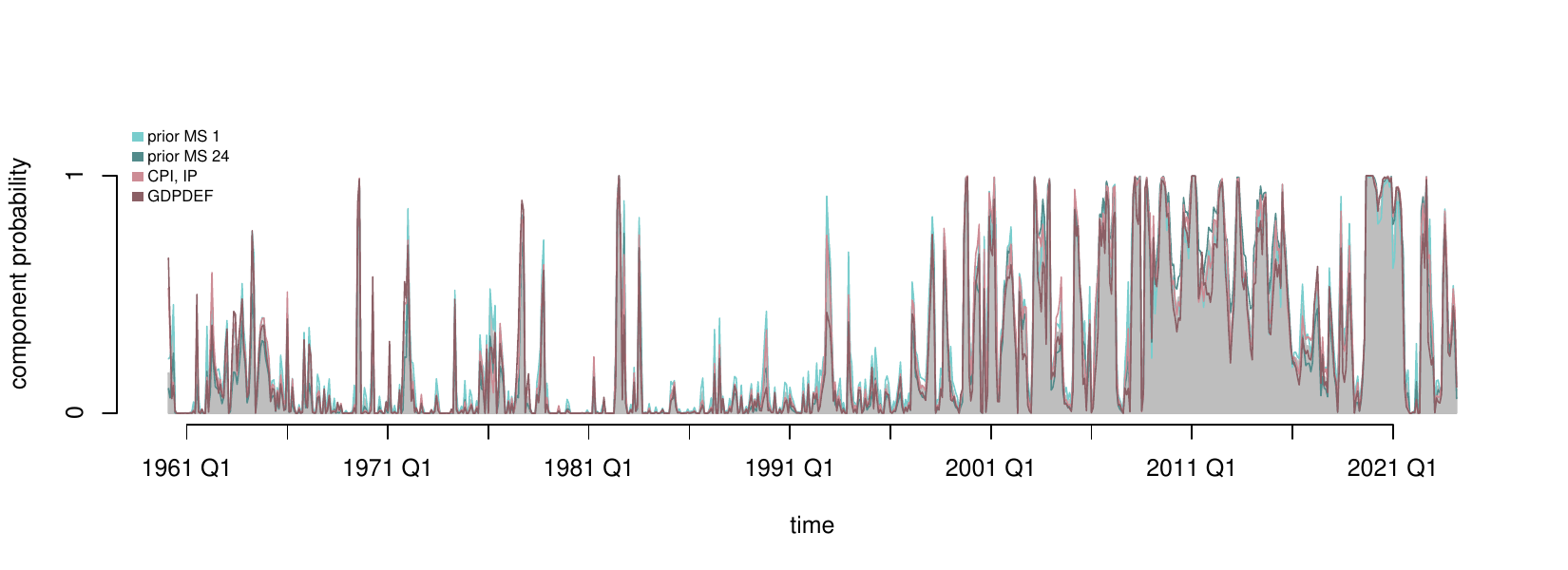} 
		\caption{Regime probabilities of the Markov process across alternative models } \label{fig:regimeprob_rob}
\end{figure}

Notably, we find that the regime allocations are robust to various changes in the model specification. The estimated regime probabilities, shown in Figure~\ref{fig:regimeprob_rob}, are very similar to a model where we impose the prior belief that regimes are less persistent or more persistent by setting the prior expected regime duration of the Markov process to one month, called \emph{prior MS 1}, or 2 years doubling the expected prior regime duration, called \emph{prior MS 24}, respectively. Moreover, models using alternative measurements for our variables, such as the consumer price index and industrial production index, model called \emph{CPI, IP}, or the interpolated GDP deflator, model called \emph{GDPDEF}, yield similar regime allocation probabilities. 

\subsection{Three results illustrating regime differences}\label{ssec:3results}


Next, we discuss differences across the regimes in more detail by looking at sample moments of data, estimates of the contemporaneous relations in the interest rate equations, and dynamic responses to US monetary policy shocks.

\paragraph{Regime-specific sample moments of data}
Regime-specific sample moments reported in Table~\ref{tab:regimeparam} show that the federal funds rate is on average lower and the term spread higher in the second regime while they both exhibit lower volatility then. Money is twice as volatile in the second regime. Moreover, the covariance between interest rates and term spreads is higher in the first than in the second regime. Differences are strong with respect to the relation between money and interest rates. In the first regime the covariance is positive, while in the second it is negative.

\begin{table}[h!]
	\caption{Regime-specific sample moments of data} \label{tab:regimeparam}
	\vspace{-0.5cm}
	\begin{center}
	\begin{tabular}{lccccccc}
		\hline\hline
& \multicolumn{3}{c}{Regime 1} && \multicolumn{3}{c}{Regime 2} \\   
& mean & sd & cov with $R_t$ &&  mean & sd &cov with $R_t$ \\ \cline{2-4} \cline{6-8} 
  $R_t$ & 5.64 & 3.62 &  && 2.54 & 2.73 &  \\ 
  $TS_t$ & 0.84 & 1.69 & $-$3.79 && 1.43 & 1.34 & $-$2.15 \\ 
  $m_t$ &  6.33 & 3.86 & 1.43 && 7.38 & 8.63 &$-$2.03 \\ 
	\hline\hline
	\end{tabular}
	\end{center}
	\footnotesize
\renewcommand{\baselineskip}{11pt} 
 \textbf{Note:} Table reports sample means, standard deviations, and the covariance with interest rates for variables in both regimes. The regime-specific moments are given for the series in first differences for $m_t$.
\end{table}

The sample regime-specific moments indicate that the second regime captures monetary policy that is characterized by cutting the interest rates and narrowing the bond yield curve supplemented by expanding the monetary basis. This supports our TVI results on the dominate role of money in the second regime. It suggest that the differences in the correlation of the variables picked up in our model by the contemporaneous relations are a decisive feature to determine regime allocation and differences across the regimes.

\paragraph{Parameter estimates of the interest rate equation}
The posterior means of the contemporaneous coefficients in the interest rate equation, the third row of the structural matrix with pattern \textit{with TS} and \textit{with m}, with the bounds of the 68\% highest density interval in parentheses,
\begin{align*}
	\text{Regime 1}&& \quad \underset{(-0.00;\ 0.01)}{0.01} y_t \underset{(-0.09;\ -0.05)}{-0.07} \pi_t   \underset{(3.45;\ 3.82)}{+3.64}  R_t  \underset{(3.92;\ 4.25)}{+4.09} TS_t &= \text{lags} + \hat{u}_t^{mps}\\
	\text{Regime 2}&& \quad \underset{(-0.03;\ -0.02)}{-0.03}y_t \underset{(-0.05;\ 0.02)}{-0.02} \pi_t   \underset{(11.32;\ 15.94)}{+14.00}  R_t  \underset{(-0.71;\ -0.02)}{-0.45} m_t &= \text{lags} + \hat{u}_t^{mps}
\end{align*} 
show notable differences across the regimes. In the first regime, interest rates and term spreads are clearly contemporaneously related. The posterior mean of the contemporaneous coefficient on term spreads is positive, with strong evidence that it is different from zero. Interest rates move simultaneously with inflation, whereas the coefficient on output is not different from zero.  

In the second regime, money aggregate plays a dominant role, as evidenced by a significant and relatively highest in absolute terms coefficient on this variable. This regime also observes an increase in the importance of output relative to inflation. The coefficient on the latter is not different from zero. This fact, together with the timing of the second regime, is in line with the shift in the emphasis in the monetary policy after 2000 described by \cite{ba2016}. For the remaining two identification pattern which get posterior weights in the second regime the posterior estimates
\begin{align*}
	\text{Regime 2:}&& \quad \underset{(-0.03;\ -0.01)}{-0.02} y_t \underset{(-0.04;\ 0.02)}{-0.01} \pi_t   \underset{(12.24;\ 16.87)}{+14.86}  R_t   &= \text{lags} + \hat{u}_t^{mps}\\
	\text{Regime 2}&& \quad \underset{(-0.03;\ -0.01)}{-0.02} y_t \underset{(-0.04;\ 0.02)}{-0.01} \pi_t     \underset{(11.93;\ 16.58)}{+14.83}  R_t \underset{(-0.45;\ 0.53)}{+0.02} TS_t &= \text{lags} + \hat{u}_t^{mps}
\end{align*} support the emphasis on output while at the same time the coefficient on term spreads is not different from zero.

\paragraph{Dynamic responses to US monetary policy shocks}

\noindent We find notable differences in the dynamic responses to the US monetary policy shock across the regimes. Based on its effects in the first regime, monetary policy can be characterized as conventionally inflation-targeting, whereas in the second one, as expansionary with specific macro-financial interactions. Our analysis is based on Figure~\ref{fig:IRF} reporting the median impulse responses to a monetary policy shock decreasing interest rates by 25 basis points in the first (light color) and second (dark color) regime together with the 68\% posterior highest density sets. We discuss the effects of an expansionary monetary policy shock.

\begin{figure}[t]
	\subfloat[response of $y$]{\includegraphics[trim={0.0cm 0.5cm 0.0cm 2.0cm},width=0.33\textwidth]{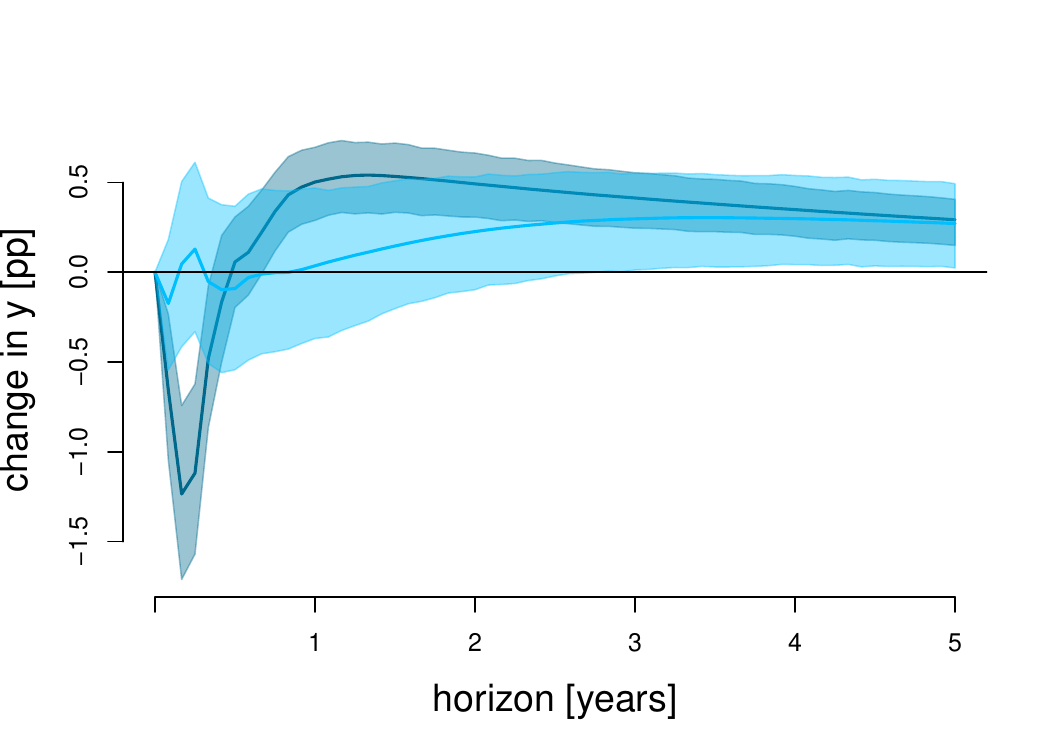}}
	\subfloat[response of $\pi$]{\includegraphics[trim={0.0cm 0.5cm 0.0cm 2.0cm},width=0.33\textwidth]{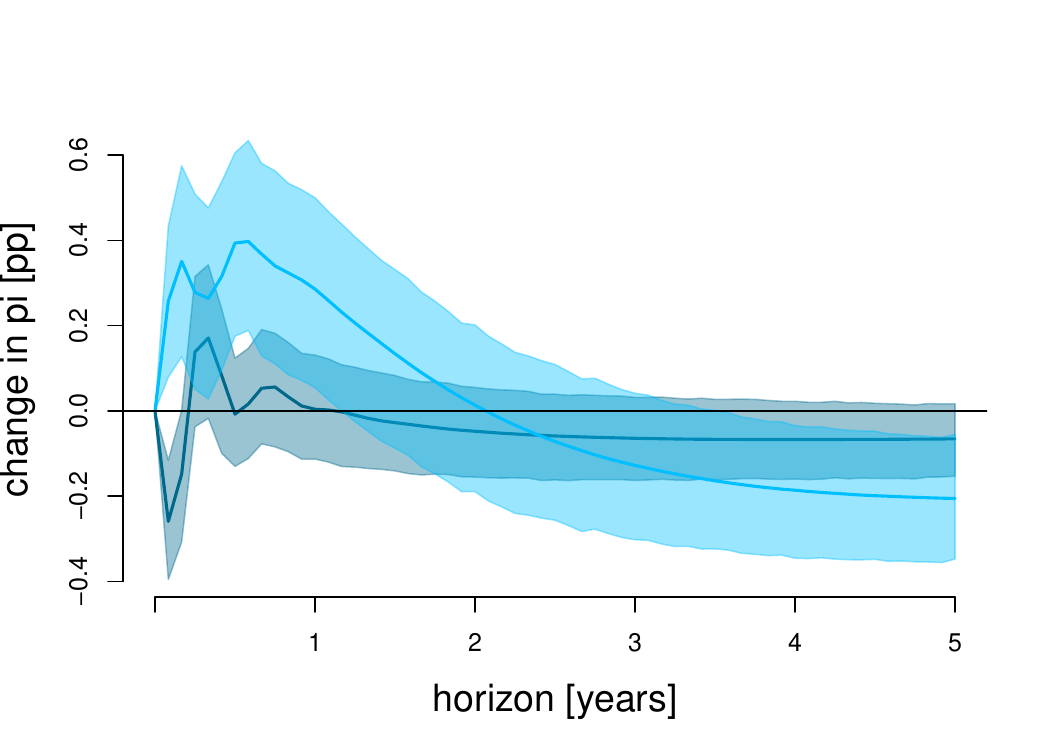}}
	\subfloat[response of $R$]{	\includegraphics[trim={0.0cm 0.5cm 0.0cm 2.0cm},width=0.33\textwidth]{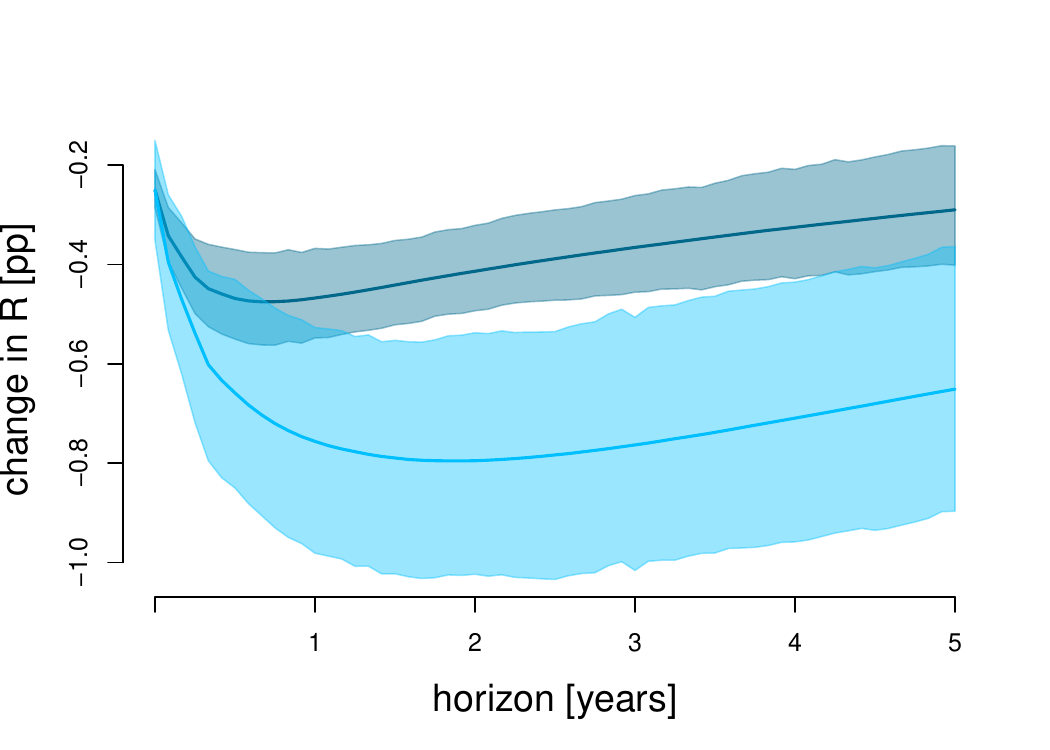}} \\
	\subfloat[response of $TS$]{	\includegraphics[trim={0.0cm 0.5cm 0.0cm 2.0cm},width=0.33\textwidth]{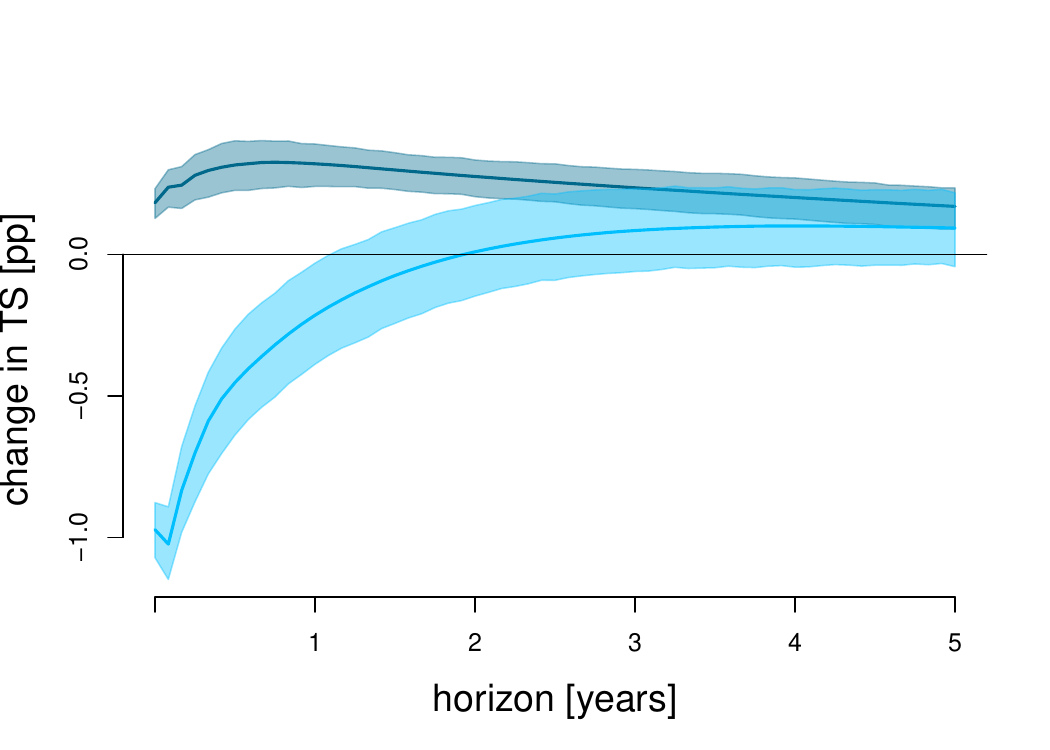}}
	\subfloat[response of $m$]{	\includegraphics[trim={0.0cm 0.5cm 0.0cm 2.0cm},width=0.33\textwidth]{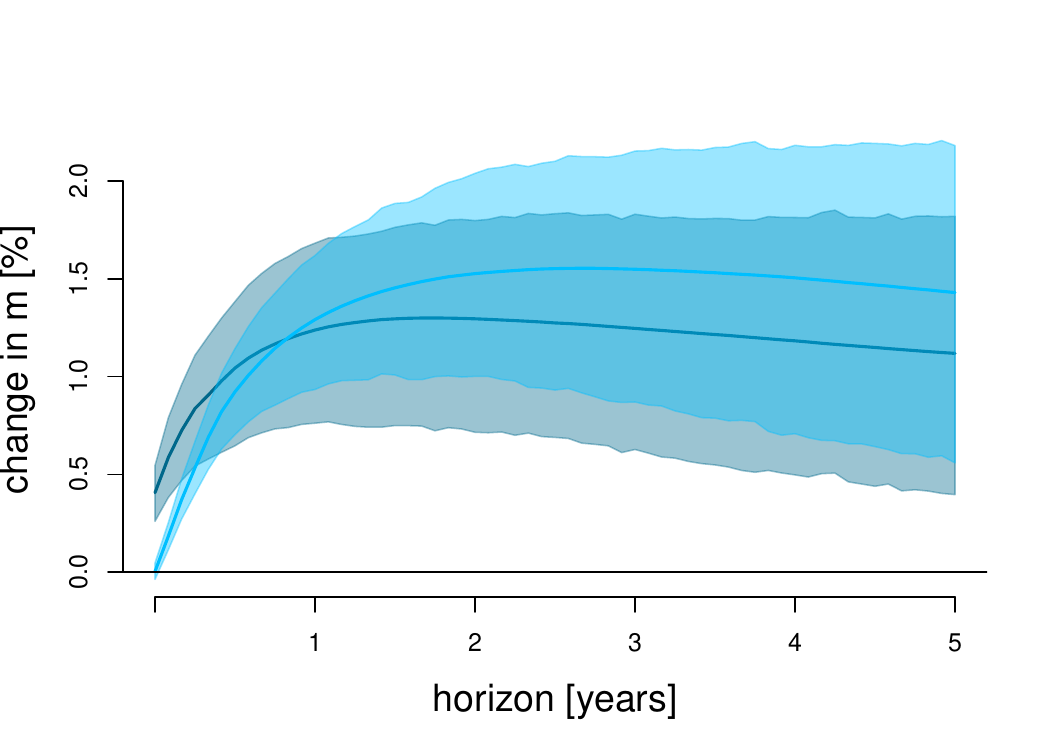}}
	\subfloat[response of $sp$]{	\includegraphics[trim={0.0cm 0.5cm 0.0cm 2.0cm},width=0.33\textwidth]{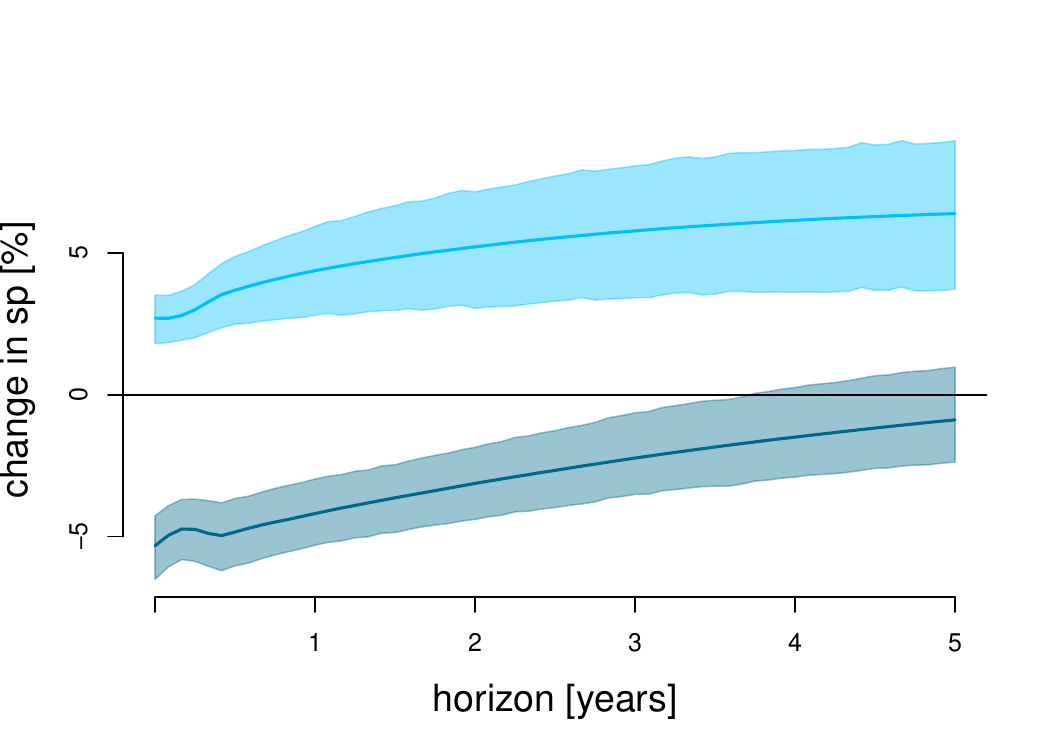}}
		\caption{Impulse response to US monetary policy shocks} \label{fig:IRF}
\end{figure}

In the first regime, an unexpected monetary loosening leads to a persistent level shift in the federal funds rate. The response falls to the lowest point after two years at 80 basis points. This shock is effective in increasing inflation with the highest impact of 0.40 pp after seven months. The effect on output growth is positive in the long run while the response is not different from zero at short horizons. 
It suggests that in the short run the inflationary push overrules potential direct effects of interest rate decreases on output growth in rather calm and economically stable periods. The system features the liquidity effect with the money supply substantially and permanently increasing. In line with \cite{Ang2011}, the lower interest rates decrease the slope of the bond yield curve for nearly a year, with the trough response narrowing the term spread by 102 basis points after one month. They also lead to a sharp increase in stock prices. 

In the second regime, a negative monetary policy shock decreases the federal funds rate permanently but slightly more moderately than in the first regime. In this regime, the monetary policy is strongly expansionary in the long run as output growth increases permanently and after a re-bouncing effect by 0.54 pp within the second year. At the same time, inflation's response is negligible. Therefore, this regime observes an increase in the importance of output relative to inflation in the monetary policy reaction function that together with the timing of the second regime, is in line with the shift in the emphasis in the monetary policy after 2000 described by \cite{ba2016}.

In these times of crises, the liquidity effect is positive. The slope of the bond yield curve is only slightly affected as the term spread widens by 32 basis points after six months, strengthening the conclusion by \cite{Tillmann2020} on time variation in this effect. The response cools off gradually but remains significant at all horizons. The stock prices contract by more than 5\% on impact. This result closely follows the change over time of the impulse responses reported by \cite{gali2015} that the authors explained by the existence of rational bubbles in the stock market.

These effects are in line with the characterisation of the monetary policy in the US as unconventional in the periods aligning with the second regime. The main objectives of these policies were to stimulate the economy in the aftermaths of the crises and provide liquidity, all of which was implemented without substantially increasing inflation in line with ``missing inflation`` during the 2010s \citep[see][]{bobeica2019}, a period which is covered by the second regime. 

Our results indicate time-variation in the monetary policy shock identification. Time-varying parameters and identification play a crucial role in the interpretability of the shocks. It is not obtained if any of these features are omitted from the model formulation. 

\section{Conclusions} \label{sec:conclusions}

\noindent We propose a new Bayesian structural vector autoregression which allows for time-varying identification. Our model facilitates a data-driven search of regime-specific exclusion restriction patterns. It features Markov-switching regimes in the structural matrix and a standard deviation parameter as well as structural shock variances following a stochastic volatility process. These components allow us to study time-variation in structural shock identification, time-varying impulse responses, and provide an easy way to check identification through heteroskedasticity within each regime as a validity condition of regime-specific identification. 

In an empirical illustration, we show that data strongly support time-variation in US monetary policy shock identification. We find evidence that during crisis and periods of unconventional monetary policy actions monetary policy shock identification is best captured by exclusion restrictions which allow money and interest rate to co-move. Before 2000 and during normal times, no restriction on the term spread is the data supported identification scheme for monetary policy shocks.

\bibliographystyle{chicago}
\bibliography{bib}



\section*{Supplementary material}

# Online Supplement: Time-Varying Identification of Structural Vector Autoregressions

Annika Camehl[a,*], Tomasz Woźniak[b]

[a]*Erasmus University Rotterdam, Burgemeester Oudlaan 50, 3062 PA Rotterdam, The Netherlands*
[b]*University of Melbourne, 111 Barry St., 3053 Carlton, VIC, Australia*

**Appendix A. Gibbs sampler**

For simplicity, we write the VAR model in compact form

$$\mathbf{y}_t = \mathbf{A}\mathbf{x}_t + \boldsymbol{\varepsilon}_t \tag{A1}$$

where the $N \times (Np+d)$-matrix $\mathbf{A} = \begin{bmatrix} \mathbf{A}_1 & \dots & \mathbf{A}_p & \mathbf{A}_d \end{bmatrix}$ collects autoregressive and constant parameters and the $(Np+d)$-vector $\mathbf{x}_t = \begin{bmatrix} \mathbf{y}'_{t-1} & \dots & \mathbf{y}'_{t-p} & \mathbf{d}'_t \end{bmatrix}'$ contains lagged endogenous variables and constant terms.

*Sampling parameters in the structural matrix*

To estimate the parameters in the structural matrix, we use a Gibbs sampler. The full conditional posterior distribution of the unrestricted elements in regime $m$, equation $n$, and TVI component $k$ is generalised-normal, proportional to

$$|\det(\mathbf{B}_{m.\mathbf{k}})|^{T_m} \exp\left\{-\frac{1}{2}\mathbf{b}_{n.m.k}\mathbf{V}_{n.m.k}\overline{\boldsymbol{\Omega}}_{B.nm}^{-1}\mathbf{V}'_{n.m.k}\mathbf{b}'_{n.m.k}\right\} \tag{A2}$$

where $\mathbf{B}_{m.\mathbf{k}}$ is the matrix from regime $m$ and TVI component $k$, with scale matrix

$$\overline{\boldsymbol{\Omega}}_{B.nm}^{-1} = \gamma_{B.n}^{-1}\mathbf{I}_N + \sum_{t:s_t=m}(\mathbf{y}_t - \mathbf{A}\mathbf{x}_t)(\mathbf{y}_t - \mathbf{A}\mathbf{x}_t)' \tag{A3}$$

---

[*]Corresponding author. *Email address:* camehl@ese.eur.nl.

1

where $T_m$ is the number of observations in regime $m$. Within each regime, we sample the structural matrix and the TVI indicators jointly.

1. Sample parameters of the structural matrix: **B** row-by-row following the algorithm by [Waggoner and Zha (2003)](#) by sampling $K$ vectors $\mathbf{b}_{n.m.k}$ from (A2) given $\kappa(m) = k$ for all $k \in \{1, \ldots, K\}$.

2. Sample TVI indicators: by computing the conditional posterior probabilities of each of these TVI components, $\overline{p}_{n.m.k}$, proportional to the product of the likelihood function, $L(\theta \mid \mathbf{Y}_T)$, and the prior distributions:

$$\overline{p}_{n.m.k} \propto L(\theta \mid \mathbf{Y}_T, \kappa(m) = k) p\left(\mathbf{B}_{n.m.k} \mid \gamma_B, \kappa(m) = k\right) p\left(\kappa(m) = k\right). \tag{A4}$$

Sample the TVI indicator from the multinomial distribution using the probabilities $\overline{p}_{n.m.k}$ for $k = 1, \ldots, K$, and return the draw from the joint posterior $(\mathbf{b}_{n.m.k}, k)$.

3. Sample structural hyper-parameters: $\gamma_{B.n}$, $\underline{s}_{B.n}$, and $\underline{s}_{\gamma_B}$ from their respective full conditional posterior distributions:

$$\gamma_{B.n} \mid \mathbf{B}_{n.m.k}, k, \underline{s}_{B.n} \sim \mathcal{IG}2\left(\underline{s}_{B.n} + \sum_m \mathbf{b}_{n.m.k_n} \mathbf{b}'_{n.m.k_n}, \quad \underline{\nu}_B + \sum_m r_{n.m.k_n}\right) \tag{A5}$$

$$\underline{s}_{B.n} \mid \gamma_B, \underline{s}_{\gamma_B} \sim \mathcal{G}\left((\underline{s}_{\gamma_B}^{-1} + (2\gamma_{B.n})^{-1})^{-1}, \underline{\nu}_{\gamma_B} + 0.5\underline{\nu}_B\right) \tag{A6}$$

$$\underline{s}_{\gamma_B} \mid \underline{s}_B \sim \mathcal{IG}2\left(\underline{s}_{s_B} + 2\sum_{n=1}^N \underline{s}_{B.n}, \underline{\nu}_{s_B} + 2N\underline{\nu}_{\gamma_B}\right) \tag{A7}$$

If TVI is not implemented in a particular row, sample $\mathbf{b}_{n.m}$ bypassing Step 2 of the algorithm.

*Sampling of the remaining parameters*

Estimation of the remaining is implemented by the Gibbs sampler. The rows of the transition probabilities matrix, **P**, and the initial state probabilities vector, $\pi_0$, are all sampled from independent $M$-variate Dirichlet full conditional posterior distributions as in [Frühwirth-Schnatter (2006)](#). The sampling procedure for the Markov process

realisations follows the forward-filtering backward-sampling algorithm proposed by Chib (1996).

*Sampling autoregressive parameters*

We follow Chan, Koop, and Yu (2024) and sample the autoregressive parameters of the $\mathbf{A}$ matrix efficiently row-by-row, where the $n^{\text{th}}$ row is denoted by $[\mathbf{A}]_{n\cdot}$. Denote by $\mathbf{A}_{n=0}$ the $\mathbf{A}$ matrix with its $n^{\text{th}}$ row set to zeros. To derive the sampler for the row of the autoregressive matrix, rewrite the structural form in an equivalent form as:

$$\mathbf{B}_{m.k}\left(\mathbf{y}_t - \mathbf{A}_{n=0}\mathbf{x}_t\right) = \left([\mathbf{B}_{m.k}]_{\cdot n} \otimes \mathbf{x}_t'\right)[\mathbf{A}]_{n\cdot}' + \boldsymbol{\varepsilon}_t \tag{A8}$$

Define an $n \times 1$ vector $\mathbf{z}_t^n = \mathbf{B}_{m.k}\left(\mathbf{y}_t - \mathbf{A}_{n=0}\mathbf{x}_t\right)$ and an $N \times (Np+d)$ matrix $\mathbf{Z}_t^n = \left([\mathbf{B}_{m.k}]_{\cdot n} \otimes \mathbf{x}_t'\right)$ and rewrite equation (A8) as:

$$\mathbf{z}_t^n = \mathbf{Z}_t^n [\mathbf{A}]_{n\cdot}' + \boldsymbol{\varepsilon}_t, \tag{A9}$$

where the shocks follow a normal distribution. This results in the multivariate normal full conditional posterior distribution for the rows of $\mathbf{A}$:

$$[\mathbf{A}]_{n\cdot}' \mid \mathbf{Y}_T, [\mathbf{A}]_{n=0}, \mathbf{B}_{m.k}, \sigma_1^2, \dots \sigma_T^2, \gamma_{A.n} \sim \mathcal{N}_{Np+d}\left(\overline{\boldsymbol{\Omega}}_{A.n}\overline{\mathbf{m}}_{A.n}', \overline{\boldsymbol{\Omega}}_{A.n}\right) \tag{A10}$$

$$\overline{\boldsymbol{\Omega}}_{A.n}^{-1} = \gamma_{A.n}^{-1}\underline{\boldsymbol{\Omega}}_A^{-1} + \sum_t \mathbf{Z}_t^{n\prime}\operatorname{diag}\left(\sigma_t^2\right)^{-1}\mathbf{Z}_t^n \tag{A11}$$

$$\overline{\mathbf{m}}_{A.n}' = \gamma_{A.n}^{-1}\underline{\boldsymbol{\Omega}}_A^{-1}\underline{\mathbf{m}}_{A.n}' + \sum_t \mathbf{Z}_t^{n\prime}\operatorname{diag}\left(\sigma_t^2\right)^{-1}\mathbf{z}_t^n \tag{A12}$$

3

Finally, the hyper-parameters of the autoregressive shrinkage hierarchical prior are sampled from their respective full conditional posterior distributions:

$$\gamma_{A.n} \mid [\mathbf{A}]_{n\cdot}, \underline{s}_{A.n} \sim \mathcal{IG}2\left(\underline{s}_{A.n} + \left([\mathbf{A}]_{n\cdot} - \underline{\mathbf{m}}_{A.n}\right)\underline{\mathbf{\Omega}}_{A.n}^{-1}\left([\mathbf{A}]_{n\cdot} - \underline{\mathbf{m}}_{A.n}\right)', \quad \underline{\nu}_A + Np + d\right) \quad (A13)$$

$$\underline{s}_{A.n} \mid \gamma_{A.n}, \underline{s}_{\gamma_A} \sim \mathcal{G}\left(\left(\underline{s}_{\gamma_A}^{-1} + (2\gamma_{A.n})^{-1}\right)^{-1}, \underline{\nu}_{\gamma_A} + 0.5\underline{\nu}_A\right) \quad (A14)$$

$$\underline{s}_{\gamma_A} \mid \underline{s}_A \sim \mathcal{IG}2\left(\underline{s}_{s_A} + 2\sum_{n=1}^{N}\underline{s}_{A.n}, \underline{\nu}_{s_A} + 2N\underline{\nu}_{\gamma_A}\right) \quad (A15)$$

*Sampling stochastic volatility parameters*

Gibbs sampler is facilitated using the auxiliary mixture sampler proposed by Omori, Chib, Shephard, and Nakajima (2007).

Specify the $n^{\text{th}}$ structural shock as:

$$u_{n.t} = \exp\left\{\frac{1}{2}\omega_n(s_t)h_{n.t}\right\}z_{n.t}, \quad (A16)$$

where $z_{n.t}$ is a standard normal innovation. Transform this equation by squaring and taking the logarithm of both sides obtaining:

$$\widetilde{u}_{n.t} = \omega_n(s_t)h_{n.t} + \widetilde{z}_{n.t}, \quad (A17)$$

where $\widetilde{u}_{n.t} = \log u_{n.t}^2$ and $\widetilde{z}_{n.t} = \log z_{n.t}^2$. The distribution of $\widetilde{z}_{n.t}$ is $\log \chi_1^2$. This non-standard distribution is approximated precisely by a mixture of ten normal distributions defined by Omori et al. (2007). This mixture of normals is specified by $q_{n.t} = 1, \ldots, 10$ – the mixture component indicator for the $n^{\text{th}}$ equation at time $t$, the normal component probability $\pi_{q_{n.t}}$, mean $\mu_{q_{n.t}}$, and variance $\sigma_{q_{n.t}}^2$. The latter three parameters are fixed and given in Omori et al. (2007), while $q_{n.t}$ augments the parameter space and is estimated. Its prior distribution is multinomial with probabilities $\pi_{q_{n.t}}$.

Define $T \times 1$ vectors: $\mathbf{h}_n = \begin{bmatrix} h_{n.1} & \ldots & h_{n.T} \end{bmatrix}'$ collecting the log-volatilities, $\mathbf{q}_n = \begin{bmatrix} q_{n.1} & \ldots & q_{n.T} \end{bmatrix}'$ collecting the realizations of $q_{n.t}$ for all $t$, $\boldsymbol{\mu}_{\mathbf{q}_n} = \begin{bmatrix} \mu_{q_{n.1}} & \ldots & \mu_{q_{n.T}} \end{bmatrix}'$, and

$\sigma^2_{\mathbf{q}_n} = \begin{bmatrix} \sigma^2_{q_{n.1}} & \dots & \sigma^2_{q_{n.T}} \end{bmatrix}'$ collecting the $n^{\text{th}}$ equation auxiliary mixture means and variances, $\widetilde{\mathbf{U}}_n = \begin{bmatrix} \widetilde{u}_{n.1} & \dots & \widetilde{u}_{n.T} \end{bmatrix}'$, and $\boldsymbol{\omega}_n = \begin{bmatrix} \omega_n(s_1) & \dots & \omega_n(s_T) \end{bmatrix}'$ collecting the volatility of the volatility parameters according to their current time assignment based on the sampled realizations of the Markov process. Whenever a subscript on these vectors is extended by the Markov process' regime indicator $m$ it means that the vector contains only the $T_m$ observations for $t$ such that $s_t = m$, e.g., $\widetilde{\mathbf{U}}_{n.m}$. Finally, define a $T \times T$ matrix $\mathbf{H}_{\rho_n}$ with ones on the main diagonal, value $-\rho_n$ on the first subdiagonal, and zeros elsewhere, and a $T_m \times T_m$ matrix $\mathbf{H}_{\rho_n.m}$ with rows and columns selected according to the Markov state allocations such that $s_t = m$. Sampling latent volatilities $\mathbf{h}_n$ proceeds independently for each $n$ from the following $T$-variate normal distribution parameterized following Chan and Jeliazkov (2009) in terms of its precision matrix $\overline{\mathbf{V}}^{-1}_{\mathbf{h}_n}$ and location vector $\overline{\mathbf{h}}_n$ as:

$$\mathbf{h}_n \mid \mathbf{Y}_T, \mathbf{s}_n, \mathbf{q}_n, \mathbf{B}, \mathbf{A}, \boldsymbol{\omega}_n, \rho_n \sim \mathcal{N}_T\left(\overline{\mathbf{V}}_{\mathbf{h}_n}\overline{\mathbf{h}}_n, \overline{\mathbf{V}}_{\mathbf{h}_n}\right) \tag{A18}$$

$$\overline{\mathbf{V}}^{-1}_{\mathbf{h}_n} = \text{diag}\left(\boldsymbol{\omega}^2_n\right)\text{diag}\left(\boldsymbol{\sigma}^{-2}_{\mathbf{q}_n}\right) + \mathbf{H}'_{\rho_n}\mathbf{H}_{\rho_n} \tag{A19}$$

$$\overline{\mathbf{h}}_n = \text{diag}\left(\boldsymbol{\omega}_n\right)\text{diag}\left(\boldsymbol{\sigma}^{-2}_{\mathbf{q}_n}\right)\left(\widetilde{\mathbf{u}}_n - \boldsymbol{\mu}_{\mathbf{q}_n}\right) \tag{A20}$$

The precision matrix, $\overline{\mathbf{V}}^{-1}_{\mathbf{h}_n}$, is tridiagonal, which greatly leads to a simulation smoother proposed by McCausland, Miller, and Pelletier (2011).

The regime-dependent volatility of the volatility parameters are sampled independently from the following normal distribution:

$$\omega_n(m) \mid \mathbf{Y}_m, \mathbf{s}_{n.m}, \mathbf{s}_{n.m}, \mathbf{h}_{n.m}, \sigma^2_{\omega_n} \sim \mathcal{N}\left(\overline{v}_{\omega_{n.m}}\overline{\omega}_{n.m}, \overline{v}_{\omega_{n.m}}\right) \tag{A21}$$

$$\overline{v}^{-1}_{\omega_{n.m}} = \mathbf{h}'_{n.m}\text{diag}\left(\boldsymbol{\sigma}^{-2}_{\mathbf{q}_n.m}\right)\mathbf{h}_{n.m} + \sigma^{-2}_{\omega_n} \tag{A22}$$

$$\overline{\omega}_n = \mathbf{h}'_{n.m}\text{diag}\left(\boldsymbol{\sigma}^{-2}_{\mathbf{q}_n.m}\right)\left(\widetilde{\mathbf{u}}_{n.m} - \boldsymbol{\mu}_{\mathbf{q}_n.m}\right) \tag{A23}$$

The autoregressive parameters of the SV equations are sampled independently from

the truncated normal distribution using the algorithm proposed by Robert (1995):

$$\rho_n \mid \mathbf{Y}_T, \mathbf{h}_n \sim \mathcal{N}\left(\left(\sum_{t=0}^{T-1} h_{n.t}^2\right)^{-1}\left(\sum_{t=1}^{T} h_{n.t} h_{n.t-1}\right), \left(\sum_{t=0}^{T-1} h_{n.t}^2\right)^{-1}\right) \mathcal{I}(|\rho_n| < 1). \quad (A24)$$

The prior variances of parameters $\omega_n(m)$, $\sigma^2_{\omega_n}$, are *a posteriori* independent and sampled from the following generalized inverse Gaussian full conditional posterior distribution:

$$\sigma^2_{\omega_n} \mid \mathbf{Y}_T, \omega_n(1), \ldots, \omega_n(M) \sim \mathcal{GIG}\left(M\underline{A} - \frac{1}{2}, \sum_m \omega_n^2(m), \frac{2}{\underline{S}}\right) \quad (A25)$$

The auxiliary mixture indicators $q_{n.t}$ are each sampled independently from a multinomial distribution with the probabilities proportional to the product of the prior probabilities $\pi_{q_{n.t}}$ and the conditional likelihood function.

Finally, proceed to the ancillarity-sufficiency interweaving sampler proposed by Kastner and Frühwirth-Schnatter (2014). Our implementation proceeds as follows: Having sampled random vector $\mathbf{h}_n$ and parameters $\omega_n(m)$, compute the parameters of the centered parameterization $\tilde{h}_{n.t} = \omega_n(m) h_{n.t}$ and $\sigma^2_{v_n} = \omega_n^2(m)$. Then, sample $\sigma^2_{v_{n.m}}$ from the generalised inverse Gaussian full conditional posterior distribution:

$$\sigma^2_{v_{n.m}} \mid \mathbf{Y}_m, \tilde{\mathbf{h}}_{n.m}, \sigma^2_{\omega_n} \sim \mathcal{GIG}\left(-\frac{T_m - 1}{2}, \tilde{\mathbf{h}}'_{n.m} \mathbf{H}'_{\rho_n.m} \mathbf{H}_{\rho_n.m} \tilde{\mathbf{h}}_{n.m}, \sigma^{-2}_{\omega_n}\right) \quad (A26)$$

using the algorithm introduced by Hörmann and Leydold (2014). Resample $\rho_n$ using the full conditional posterior distribution from equation (A24) where the vector $\mathbf{h}_n$ is replaced by $\tilde{\mathbf{h}}_n$. Finally, compute $\omega_n(m) = \pm\sqrt{\sigma^2_{v_{n.m}}}$ and $h_{n.t} = \frac{1}{\omega_n(m)} \tilde{h}_{n.t}$ and return them as the MCMC draws for these parameters.

**Appendix B. Data**

$R_t$ and $TS_t$ are expressed in annual terms, as is $\pi_t$ computed as the logarithmic rates of returns on the consumer price index expressed in percent. $y_t$ is taken in log-difference multiplied by 100. $m_t$, and $sp_t$ are taken in log-levels multiplied by 100. Series $R_t$, $y_t$, and

$m_t$ are downloaded from Board of Governors of the Federal Reserve System (2023a,b,c) respectively, $\pi_t$ is provided by Organization for Economic Co-operation and Development (2023), $TS_t$ by Federal Reserve Bank of St. Louis (2023), and $sp_t$ by yahoo finance (2023). Additionally, we used a series of shadow rates by Wu and Xia (2016) downloaded from https://sites.google.com/view/jingcynthiawu/shadow-rates [accessed 15.08.2023].

**Appendix C. Forecasting performance 12 periods ahead**

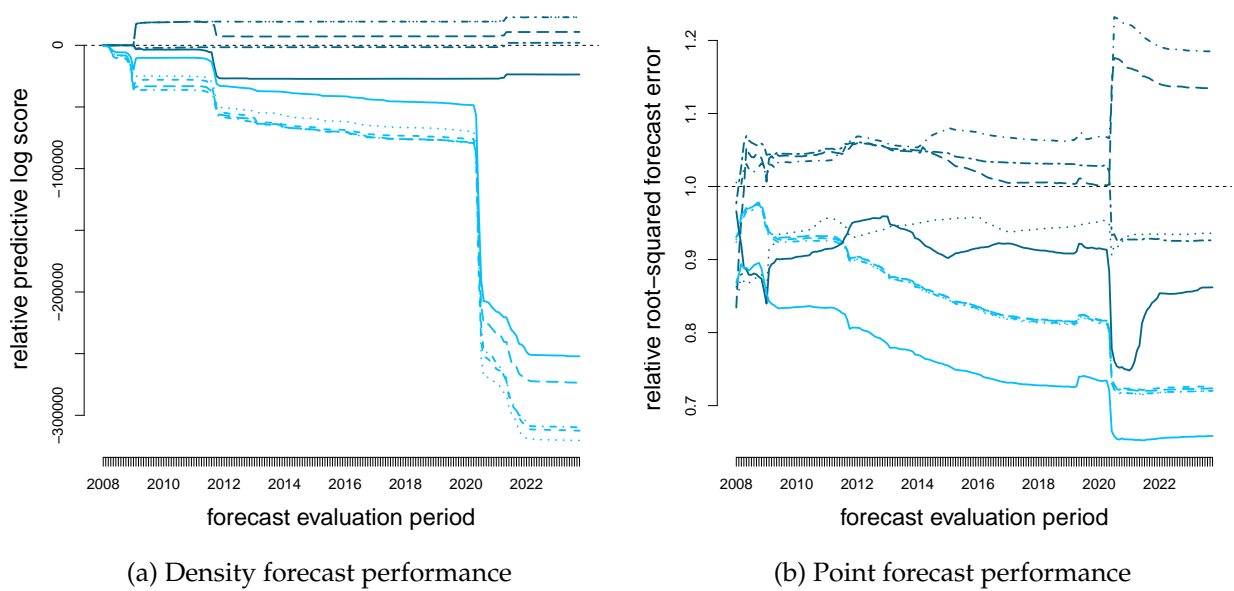

(a) Density forecast performance     (b) Point forecast performance

Figure Appendix C.1: Twelve-step-ahead forecasting performance of selected models
**Note:** The legend for models includes: MS models (dark colour): —— $M = 3$ with TVI, - - $M = 2$ baseline, ···· $M = 2$ with $TS$, -·- $M = 2$ with $m$, — — $M = 2$ only $m$, no MS ($M = 1$, bright colour): —— with TVI, - - baseline, ···· with $TS$, -·- with $m$, — — only $m$. Measures are relative to the benchmark model with $M = 2$ and TVI.

While the performance at at twelve horizons supports the benchmark model, the picture is more nuanced. The model with two regimes and TVI outperforms all alternatives except one based on both measures. The exception is the model with two regimes and restrictions allowing for contemporaneous movement of interest rates with output, inflation, and term spreads in the third row of **B** (dark dotted line). This model, however, performs weak on short horizon forecasts. Log scores across models show that

7

MS clearly improves forecasting while due to the higher uncertainty accompanying longer horizon forecasts the extra flexibility of TVI leads to small losses in forecasting accuracy. The evaluation based on point forecasts, however, provides evidence that models with TVI improve forecasting performance.

**Appendix D. Robustness**

*Appendix D.1. Time-varying identification*

We show robustness on the posterior probabilities of the TVI components in Table D.1, with probabilities reported for our main model, labeled *benchmark*, and eight alternative models in rows for the first (upper panel) and second regime (lower panel) for TVI in the third row of **B** (left part) and fourth row (right part). Selecting two different identification patterns over time is remarkably robust with respect to extending or limiting the set of possible exclusion restrictions in the search. First, broadening the set of exclusion restrictions in the third row leads to the same posterior probabilities of the TVI mechanism for the monetary policy shock identification – shown for a model which extends the restriction patterns of the benchmark model by allowing for interest rates to move simultaneously with stock prices, called *with sp*. Expanding the restriction patterns in the forth row does not change the conclusion on the support of zero exclusion restriction on interest rates in the forth row, see the model called *exp 4eq* which also includes patterns with zero restrictions on money and stock prices ("LT") or zero restrictions on interest rates, money and stock prices ("LT0").

Second, excluding the variant *only m*, does not change the posterior probability of the remaining components, see model *ex only m*. Limiting the set of possible exclusion restrictions in the third row to simultaneous movements of interest rates, output, and inflation with term spread or money supports the role of the term spread and money in the respective regimes, model called *with TS, m*.

Moreover, models with changes in the prior set-ups on the contemporaneous relations by cutting the prior standard deviation of **B** in half or doubling the prior variance, models

8

Table D.1: Posterior probabilities of the TVI components across alternative models

| | **Regime 1** | | | | | | | | |
|---|---|---|---|---|---|---|---|---|---|
| | | | *third row* | | | | *fourth row* | | |
| | base | with TS | with m | with sp | only m | UR | with 0 | LT | LT0 |
| benchmark | 0.00 | 1.00 | 0.00 | | 0.00 | 1.00 | 0.00 | | |
| plus sp | 0.00 | 1.00 | 0.00 | 0.00 | 0.00 | 1.00 | 0.00 | | |
| exp 4eq | 0.00 | 1.00 | 0.00 | | 0.00 | 0.49 | 0.00 | 0.51 | 0.00 |
| ex only m | 0.00 | 1.00 | 0.00 | | | 1.00 | 0.00 | | |
| with TS, m | | 1.00 | 0.00 | | | 1.00 | 0.00 | | |
| sd(B) /2 | 0.00 | 1.00 | 0.00 | | 0.00 | 1.00 | 0.00 | | |
| Var(B) x2 | 0.00 | 1.00 | 0.00 | | 0.00 | 1.00 | 0.00 | | |
| CPI, IP | 0.00 | 1.00 | 0.00 | | 0.00 | 1.00 | 0.00 | | |
| GDPDEF | 0.00 | 1.00 | 0.00 | | 0.00 | 1.00 | 0.00 | | |
| | **Regime 2** | | | | | | | | |
| benchmark | 0.27 | 0.27 | 0.45 | | 0.01 | 0.51 | 0.49 | | |
| plus sp | 0.16 | 0.15 | 0.38 | 0.25 | 0.05 | 0.58 | 0.42 | | |
| exp 4eq | 0.28 | 0.26 | 0.44 | | 0.02 | 0.55 | 0.44 | 0.00 | 0.00 |
| ex only m | 0.25 | 0.25 | 0.50 | | | 0.54 | 0.46 | | |
| with TS, m | | 0.42 | 0.58 | | | 0.55 | 0.45 | | |
| sd(B) /2 | 0.26 | 0.24 | 0.47 | | 0.03 | 0.54 | 0.46 | | |
| Var(B) x2 | 0.27 | 0.27 | 0.45 | | 0.01 | 0.51 | 0.49 | | |
| CPI, IP | 0.26 | 0.25 | 0.49 | | 0.01 | 0.52 | 0.48 | | |
| GDPDEF | 0.29 | 0.28 | 0.43 | | 0.00 | 0.51 | 0.49 | | |

**Note:** Table shows the posterior probabilities of $\kappa_n(m)$ for the identification of the monetary policy shock for regime one (upper panel) and regime two (lower panel) for nine different models.

called *sd(B) /2* and *Var(B) x2*, respectively, support time-variation in the identification of the monetary policy shock. Also, changes in the variables (using the consumer price index and industrial production index instead of the growth rates, model called *CPI, IP*, or the interpolated GDP deflator instead of CPI inflation, model called *GDPDEF*) yield similar TVI probabilities.

*Appendix D.2. Identification through heteroskedasticity*

We find that the identification of the third and fourth row through heteroskedasticity is robust to changing the scale hyper-parameter of the prior controlling for the regime-dependent standard deviation parameter level to 0.5 or 2, see models called *prior*

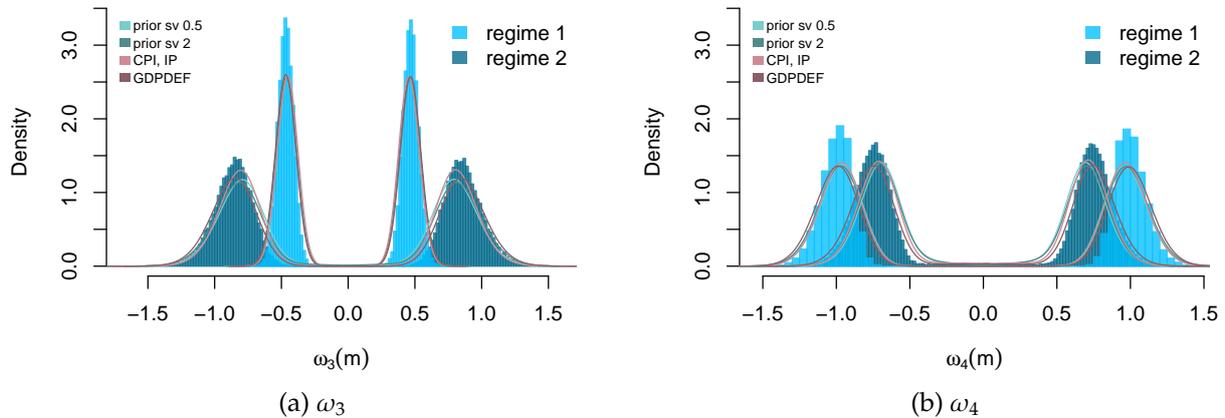

Figure Appendix D.1: Marginal posterior densities of the standard deviations across alternative models

*sv 0.5* and *prior sv 2* in Figure Appendix D.1. Also, using alternative measurements for our variables, such as the consumer price index, industrial production, or the interpolated GDP deflator, has no impact on identification through heteroskedasticity, see models called *CPI, IP* and *GDPDEF*.

# References

Board of Governors of the Federal Reserve System (2023a). Federal Funds Effective Rate [DFF]. Data. Retrieved from FRED database [accessed August 15, 2023].

Board of Governors of the Federal Reserve System (2023b). Industrial Production: Total Index [INDPRO]. Data. Retrieved from FRED database [accessed August 15, 2023].

Board of Governors of the Federal Reserve System (2023c). M2 [M2SL]. Data. Retrieved from FRED database [accessed August 15, 2023].

Chan, J. and I. Jeliazkov (2009). Efficient Simulation and Integrated Likelihood Estimation in State Space Models. *International Journal of Mathematical Modelling and Numerical Optimisation 1*, 101–120.

Chan, J. C. C., G. Koop, and X. Yu (2024). Large order-invariant bayesian vars with stochastic volatility. *Journal of Business & Economic Statistics 42*(2), 825–837.

Chib, S. (1996). Calculating Posterior Distributions and Modal Estimates in Markov Mixture Models. *Journal of Econometrics 75*(1), 79–7.

Federal Reserve Bank of St. Louis (2023). 10-Year Treasury Constant Maturity Minus Federal Funds Rate [T10YFFM]. Data. Retrieved from FRED database [accessed August 15, 2023].

Frühwirth-Schnatter, S. (2006). *Finite Mixture and Markov Switching Models*. New York: Springer Science+Business Media, LLC.

Hörmann, W. and J. Leydold (2014). Generating generalized inverse gaussian random variates. *Statistics and Computing 24*(4), 547–557.

Kastner, G. and S. Frühwirth-Schnatter (2014). Ancillarity-sufficiency interweaving strategy (asis) for boosting mcmc estimation of stochastic volatility models. *Computational Statistics & Data Analysis 76*, 408 – 423.

McCausland, W. J., S. Miller, and D. Pelletier (2011). Simulation smoothing for state–space models: A computational efficiency analysis. *Computational Statistics & Data Analysis 55*(1), 199–212.

Omori, Y., S. Chib, N. Shephard, and J. Nakajima (2007). Stochastic Volatility with Leverage: Fast and Efficient Likelihood Inference. *Journal of Econometrics 140*(2), 425–449.

Organization for Economic Co-operation and Development (2023). Main Economic Indicators: Consumer Price Index: All Items: Total for United States [CPALTT01USM661S]. Data. Retrieved from FRED database [accessed August 15, 2023].

Robert, C. P. (1995). Simulation of truncated normal variables. *Statistics and computing 5*(2), 121–125.

Waggoner, D. F. and T. Zha (2003). A Gibbs sampler for structural vector autoregressions. *Journal of Economic Dynamics and Control 28*(2), 349–366.

Wu, J. C. and F. D. Xia (2016). Measuring the macroeconomic impact of monetary policy at the zero lower bound. *Journal of Money, Credit and Banking 48*(2-3), 253–291.

yahoo finance (2023). S&P 500 ˆGSPC. Data. [accessed August 15, 2023].



\end{document}